\documentclass[10pt, conference]{IEEEtran}
\IEEEoverridecommandlockouts
\usepackage{amsmath,amssymb,amsfonts,amsthm}
\usepackage{subcaption}
\usepackage{color}
\usepackage{soul}
\usepackage[linesnumbered,ruled,vlined]{algorithm2e}
\DeclareMathOperator*{\argmax}{arg\,max}

\usepackage{algpseudocode}

\newtheorem{theorem}{Theorem}

\newtheorem{definition}{Definition}
\newtheorem{lemma}{Lemma}

\newtheorem{assumption}{Assumption}
\newtheorem{remark}{Remark}
\newtheorem{example}{Example}

\def\blue{\color{black}}

\usepackage{tcolorbox}
\usepackage{lipsum}
\usepackage{hyperref}
\usepackage[noadjust]{cite}
\usepackage{enumitem}
\def\BibTeX{{\rm B\kern-.05em{\sc i\kern-.025em b}\kern-.08em
    T\kern-.1667em\lower.7ex\hbox{E}\kern-.125emX}}

\allowdisplaybreaks
\let\hide\iffalse
 
\colorlet{blue}{blue}
\begin{document}
\title{AoI-based Scheduling of Correlated Sources \\ for Timely Inference\\
}

\IEEEoverridecommandlockouts
\author{Md Kamran Chowdhury Shisher,~\IEEEmembership{Member,~IEEE,}
        Vishrant Tripathi,~\IEEEmembership{Member,~IEEE,}\\ Mung Chiang,~\IEEEmembership{Fellow,~IEEE,} Christopher G. Brinton,~\IEEEmembership{Senior Member,~IEEE}
        \IEEEcompsocitemizethanks{\IEEEcompsocthanksitem An abridged version of this paper was presented at IEEE ICC 2025 \cite{TechICC}. 

        M.K.C. Shisher, V. Tripathi, M. Chiang, and C. Brinton are with the Elmore Family School of Electrical and Computer Engineering, Purdue University, West Lafayette, IN 47907, USA (e-mail: mshisher@purdue.edu, tripathv@purdue.edu, chiang@purdue.edu, cgb@purdue.edu).  
         
        M. K. C. Shisher and C. Brinton were supported in part by the Office of Naval Research (ONR) under grants N00014-23-C-1016 and N00014-22-1-2305, and by the National Science Foundation (NSF) under grant CPS-2313109.}}
\maketitle

\begin{abstract}
    We investigate a real-time remote inference system where multiple correlated sources transmit observations over a communication channel to a receiver. The receiver utilizes these observations to infer multiple time-varying targets. Due to limited communication resources, the delivered observations may not be fresh. To quantify data freshness, we employ the Age of Information (AoI) metric. To minimize the inference error, we aim to design a signal-agnostic scheduling policy that leverages AoI without requiring knowledge of the actual target values or the source observations. This scheduling problem is a restless multi-armed bandit (RMAB) problem with a non-separable penalty function. Unlike traditional RMABs, the correlation among sources introduces a unique challenge: the penalty function of each source depends on the AoI of other correlated sources, preventing the problem from decomposing into multiple independent Markov Decision Processes (MDPs), a key step in applying traditional RMAB solutions. To address this, we propose a novel approach that approximates the penalty function for each source and establishes an analytical bound on the approximation error. We then develop scheduling policies for two scenarios: (i) full knowledge of the penalty functions and (ii) no knowledge of the penalty functions. For the case of known penalty functions, we present an upper bound on the optimality gap that highlights the impact of the correlation parameter and the system size. For the case of unknown penalty functions and signal distributions, we develop an online learning approach that utilizes bandit feedback to learn an online Maximum Gain First policy. Simulation results demonstrate the effectiveness of our proposed policies in minimizing inference error and achieving scalability in the number of sources.
\end{abstract}


\section{Introduction}

Next-generation communications (Next-G) (e.g., 6G) are expected to support many intelligent applications such as environmental forecasting, surveillance, networked control of robot or UAV swarms, communication between connected vehicles, and massive sensing via the Internet of Things (IoTs) \cite{giordani2020toward}. These applications often require timely inference of dynamic targets (e.g., positions of moving objects, changes in environmental conditions).

In this paper, we investigate timely inference of multiple time-varying targets based on observations collected from remote sources (e.g., sensors, cameras, IoT devices, UAVs). These observations are transmitted to a receiver over a capacity-limited communication channel. Furthermore, the source observations can exhibit correlation in their dynamics. For example, in environmental monitoring, the data collected by geographically close sensors measuring the temperature or humidity would be correlated. In scenarios where vehicles communicate with other neighboring vehicles, data collected from them can be spatially correlated. Similarly, in target tracking with UAV-mounted cameras, observations from cameras with overlapping fields of view are correlated. Timely and effective scheduling of correlated sources in wireless communication networks is crucial for minimizing inference errors of dynamic targets and improving real-time performance.

Due to limited communication resources, observations delivered from remote sources may not be fresh. \emph{Age of Information} (AoI), introduced in \cite{song1990performance, kaul2012real}, provides a convenient measure of information freshness regarding the sources at the receiver. Specifically, consider packets sent from a source to a receiver: if $U(t)$ is the generation time of the most recently received packet by time $t$, then the AoI at time $t$ is the difference between $t$ and $U(t)$. Recent works on remote inference \cite{shisher2021age, shishertimely, shisher2023learning,  shisher2024AR} have shown that the inference errors for different tasks can be expressed as functions of AoI. Additionally, AoI can be readily tracked, making it a promising metric for determining how to prioritize resource allocation. Recent works \cite{tripathi2022optimizing,vishrantcorrelated} have also shown that AoI-based scheduling is sufficient to obtain near-optimal scheduling of correlated Gauss–Markov sources with linear time-invariant (LTI) system models. While this is promising, many real-world systems exhibit non-linear dynamics and complex correlation structures. Motivated by this, in this work, we pose the following research question:
\begin{quote}
\textbf{\textit{How can we develop AoI-based scheduling of correlated sources with arbitrary correlation structure to minimize inference errors for target processes?}}
\end{quote}

\subsection{Outline and Summary of Contributions}
In answering the above question, we make the following contributions:
\begin{itemize}[leftmargin=4mm]
    \item To minimize the discounted sum of inference errors for multiple targets, we formulate the problem of scheduling correlated sources over a capacity-limited channel. For the set of all causal and signal-agnostic scheduling policies in which (i) the scheduler makes its decision based on the current and the past information available and (ii) the scheduler does not have access to the realization of the actual processes, we show that the problem can be expressed as minimizing the discounted sum of AoI penalty functions (see Lemma \ref{lemma1} and \eqref{Multi-scheduling_problem1}-\eqref{Scheduling_constraint2}).   
    
    \item The scheduling problem is a restless multi-armed bandit (RMAB) problem with a non-separable penalty function. In contrast to the traditional RMAB framework \cite{whittle1988restless}, the penalty functions of all arms in this problem are intertwined. This is because the penalty function for each arm depends on the AoI values of all sources. This interconnectedness prevents the decomposition of the problem into independent Markov Decision Processes (MDPs) even with constraint relaxation, a common approach used for traditional RMABs \cite{whittle1988restless}. We address this problem by establishing an information-theoretic lower bound (see Lemma \ref{approx}) on inference error and use it to decompose the problem into multiple independent MDPs. Prior work \cite{vishrantcorrelated} also provided a lower bound of inference error under correlated settings. However, the bound in \cite{vishrantcorrelated} was limited to Gauss-Markov sources with an LTI system model. Our information-theoretic lower bound extends the bound of \cite{vishrantcorrelated} to a significantly broader class of systems, including non-linear and non-Gaussian sources. 
    
    \item After establishing the approximated function, we develop a Maximum Gain First (MGF) policy (see Algorithm \ref{alg:MGF}) when the scheduler has full knowledge of penalty functions and/or the distribution of targets and source observations. {\blue We present an upper bound on the optimality gap of our MGF policy (see Theorem \ref{theorem2}). The optimality gap highlights the impact of correlation parameter and the system size.}  

    \item For unknown penalty functions and arbitrary correlation settings, it is not possible to use the bound developed for the known penalty function. Towards this end, we first solve a Lagrangian relaxed problem with a new method of function approximation. Then, we provide an Online Threshold Policy for solving the relaxed problem in Algorithm \ref{alg:threshold} by using a bandit feedback structure. Theorem \ref{convergence} analyzes the behavior of Algorithm \ref{alg:threshold}. 
    
    \item By using the structure of the Online Threshold Policy (Algorithm \ref{alg:threshold}) and Gain index (Definition \ref{gainindex}), we provide a novel ``Online Maximum Gain First" (Online-MGF) policy in Algorithm \ref{alg:gain}. We conduct simulations (see Section \ref{simulation}) by using a correlation model developed in \cite{tripathi2022optimizing}. The simulation results show the effectiveness of our MGF and Online-MGF policies without knowing the exact correlation structure.
\end{itemize}

\subsection{Related Works}
Over the past decade, there has been a rapidly growing body of research on analyzing AoI for queuing systems \cite{kaul2012real,sun2017update, huang2015optimizing, bedewy2019minimizing, yates2015lazy, jiang2019status}, using AoI as a metric for scheduling policies in networks \cite{kadota2018optimizing, kadota2018scheduling}, for monitoring or controlling systems over networks \cite{klugel2019aoi, Tripathi2019}, and for optimizing remote estimation \cite{pan2023sampling,SunTIT2020,orneeTON2021, ornee2023whittle, OrneeMILCOM} and inference systems \cite{ShisherMobihoc, shishertimely, shisher2021age, shisher2023learning, shisher2024AR, ari2023goal, shisher2025computation}. For detailed surveys of AoI literature see \cite{yates2021age, sun2022age}.

One of the first works to analyze how AoI values of correlated source observations affect remote inference was \cite{shisher2021age}, in which an information-theoretic analysis was provided to understand the impact of AoI of correlated sources on inference error. However, the authors of \cite{shisher2021age} did not provide any scheduling policies to optimize inference performance under correlation. In \cite{tripathi2022optimizing}, the authors studied the scheduling of correlated sources based on AoI values. This work considers a probabilistic model of correlation between the sources, and then proposed an AoI-based scheduling policy for improving information freshness at the receiver. However, the probabilistic model does not reflect the impact of correlation on inference error or estimation of targets. In \cite{vishrantcorrelated}, the authors developed an AoI-based scheduling policy for minimizing estimation error of correlated Gauss-Markov sources with an LTI system model. Our current work is related to these prior works on AoI-based scheduling of correlated sources \cite{tripathi2022optimizing, vishrantcorrelated}. In particular, this paper generalizes the work of \cite{tripathi2022optimizing, vishrantcorrelated} to arbitrary systems, loss functions, and target/source processes. In addition, the prior works \cite{tripathi2022optimizing, vishrantcorrelated} consider that the distribution of the signal processes and the penalty functions are known to the scheduler. In this paper, we extend our scheduling policy to the setting, where the distribution of the signal processes and the penalty functions are not available to the scheduler.

The problem of scheduling multiple sources to optimize linear and non-linear functions of AoI can be formulated as a restless multi-armed bandit (RMAB) problem \cite{Tripathi2019, ShisherMobihoc, xiong2022index, zou2021minimizing, chen2021scheduling, chen2021uncertainty}.
While Whittle index policies \cite{Tripathi2019, ShisherMobihoc, whittle1988restless, shishertimely, chen2021uncertainty} and gain index policies \cite{OrneeMILCOM, chen2023index, shisher2023learning, shisher2025computation} have been shown to achieve good performance and asymptotic optimality in certain RMABs \cite{verloop2016asymptotically, gast2021lp}, they cannot be directly applied to our problem. This is because the correlation among sources introduces a non-separable penalty function, where the penalty for each arm depends on the AoI of other arms. This interdependence prevents the decomposition of the problem into multiple independent MDPs, a key step in applying traditional RMAB solutions. Our paper addresses this challenge by developing a novel approach to handle correlated sources with a non-separable penalty function.

Furthermore, our paper contributes to the design of online policies for RMAB problems. While existing works like \cite{avrachenkov2022whittle, wang2023optimistic, killian2021q} have explored online policies, they focus on scenarios with independent penalty functions and arms. A recent study \cite{raman2024global} investigates RMAB problems with non-separable, monotonic, and sub-modular global penalty function, in which the authors develop an approximate Whittle index-based policy. In contrast, our work considers arbitrary penalty functions and employs a gain index-based approach.  Moreover, we focus on a signal-agnostic scheduler and provide scheduling policies for both known and unknown penalty functions, unlike \cite{raman2024global}, which assumes full system information. 

\section{System Model}

Consider $M$ sources communicating over a shared wireless channel to a receiver (Fig.~\ref{fig:scheduling}). We assume discrete slotted time and at every time slot $t$, each source $m$ observes a time-varying signal $X_{m,t} \in \mathcal{X}_m$, where $\mathcal{X}_m$ represents the finite set of possible observations from source $m$. 
A scheduler progressively schedules source observations to the receiver. At each time slot $t$, the receiver uses the received observations to infer $M$ targets $(Y_{1, t}, Y_{2,t}, \ldots, Y_{M,t})$ of interest. Each target $Y_{m,t}$ is drawn from a finite set $\mathcal Y_m$ and can be directly inferred from the corresponding source observation $X_{m, t}$. {\blue For example, $Y_{m,t}$ can be a function of $X_{m, t}$, such as $Y_{m, t}=X_{m, t}$.} Moreover, the target $Y_{m,t}$ can be correlated with the observation $X_{n, t}$ of any other source $n\neq m$. The results of this paper are derived without assuming a specific correlation model. This ensures the general applicability of our results. An information theoretic interpretation on the impact of the correlation is provided in Sec. \ref{Interpretation}.  

\subsection{Communication Model}
Due to interference and bandwidth limitations, we assume that at most $N$ out of $M$ sources can be scheduled to transmit their observations at any time slot $t$, where $0<N<M$. If source $m$ is scheduled at time $t$, it transmits the current observation $X_{m, t}$ to the receiver. For simplicity, we assume reliable channels, i.e., observations sent at time slot $t$ are delivered error-free at time slot $t+1$. Because of communication constraints, the receiver may not have fresh observations from all sources. Let $X_{m, t-\Delta_m(t)}$ be the most recently delivered signal observation from source $m$ which was generated $\Delta_m(t)$ time-slots ago. We call $\Delta_{m}(t)$ the age of information (AoI) \cite{kaul2012real, shishertimely} of source $m$. Let $U_{m}(t)$ be the generation time of the most recent delivered observation from source $m$. Then, the AoI can be formally defined as:
\begin{align}\label{AoI}
    \Delta_{m}(t)=t-U_{m}(t).
\end{align}
Let $\pi_{m}(t) \in \{0, 1\}$ be the scheduling decision of source $m$. At time slot $t$, if $\pi_{m}(t) = 1$, the $m$-th source is scheduled to transmit its observation to the receiver; otherwise, if $\pi_{m}(t) = 0$, this transmission does not occur. If $\pi_m(t-1)=0$, AoI $\Delta_m(t)=\Delta_m(t-1)+1$ grows by $1$; otherwise, if $\pi_m(t-1)=1$, AoI drops to $\Delta_m(t)=1$.

\subsection{Inference Model}
The receiver employs $M$ predictors. The $m$-th predictor $\phi_m$ uses the most recent observations from all sources along with their corresponding AoI values, to produce an inference result $a_{m, t} \in \mathcal A_m$ for target $Y_{m, t}$. Specifically, given the most recently delivered source observations and their AoI values
$$(X_{m, t-\Delta_m(t)}, \Delta_{m}(t))_{m=1}^M=(x_{m}, \delta_{m})_{m=1}^M,$$ the inference result $$a_{m,t}=\phi_m((x_{m}, \delta_{m})_{m=1}^M)\in \mathcal A_m$$ minimizes the expected loss function $\mathbb E[L(Y_{m, t}, a_m)]$ over all possible inference results $a_m \in \mathcal A_m$.
$L(y_m, a_m)$ is the loss incurred when the actual target is $Y_{m,t}=y_m$ and the predicted output is $a_{m,t}=a_m$. The loss function $L(
\cdot, \cdot)$ and the output space $\mathcal A_m$ can be designed according to the goal of the system. For example, $\mathcal A_m=\mathcal Y_m$ and quadratic loss $\|y-\hat y\|^2$ can be used for a regression task. In maximum likelihood estimation, we can use logarithmic loss function and the output space $\mathcal A_m$ can be $\mathcal P_{\mathcal Y_{m}}$, which is the set of all probability distributions on $\mathcal Y_m$. At time $t$, the expected inference error for target $Y_{m, t}$ is given by
$\mathbb E\left[L(Y_{m, t}, a_{m,t})\right].$

\begin{figure}[t]
\centering
\includegraphics[width=0.40\textwidth]{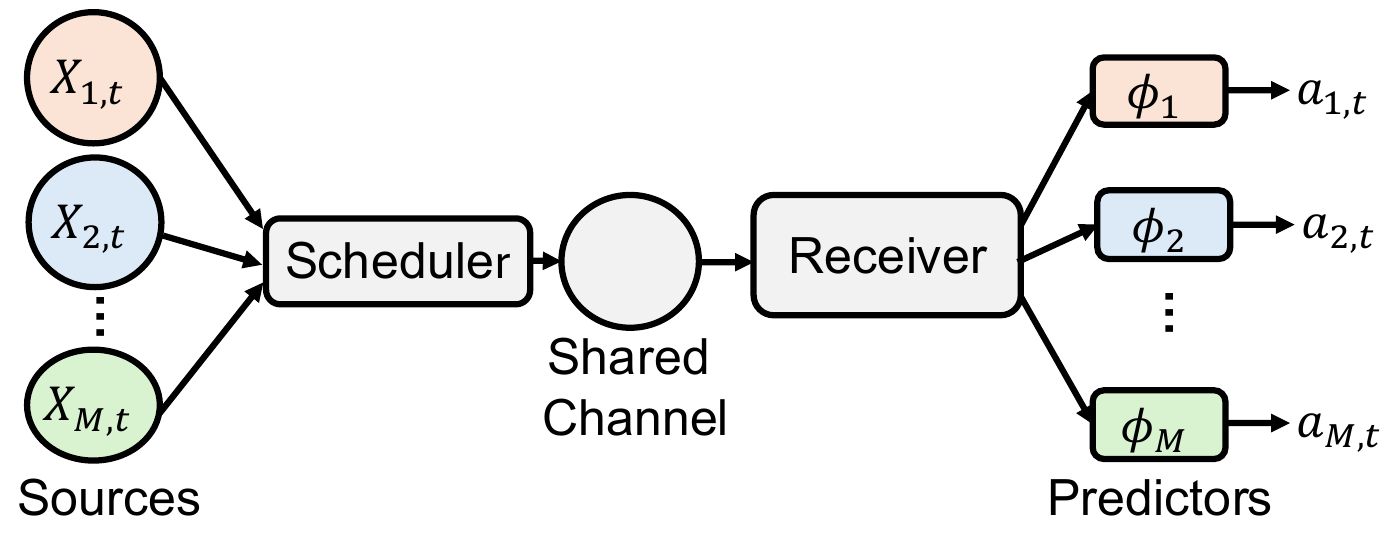}
\caption{\small A remote inference system with $M$ correlated sources, scheduler, shared channel, receiver, and $M$ predictors. 
\label{fig:scheduling}
}
\vspace{-3mm}
\end{figure}

\section{Problem Formulation}
We focus on a class of \emph{signal-agnostic} scheduling policies, where scheduling decisions are made without knowledge of the observed process's current signal values, i.e., at time $t$ the centralized scheduler does not have access to $\{(Y_{m, t}, X_{m, t}), m=1, \ldots, M, t=0, 1, \ldots\}.$ Moreover, we make the following assumption on the observed and the target processes:

\begin{assumption}\label{stationary}
    The process $$\{Y_{m,t}, X_{1,t}, X_{2, t}, \ldots, X_{M, t}, t\in \mathbb Z\}$$ is jointly stationary, i.e., the joint distributions of $$(Y_{m, t}, X_{1,t_1}, X_{2, t_2}, \ldots, X_{M, t_M})$$ and $$(Y_{m, t+\tau}, X_{1,t_1+\tau}, X_{2, t_2+\tau}, \ldots, X_{M, t_M+\tau})$$ are same for all $\tau, t, t_1, t_2, \ldots, t_M$. Moreover, the assumption holds for all target process $Y_{m, t}$ with $m=1,2, \ldots, M$.
\end{assumption}

Assumption \ref{stationary} implies that the dependence of the inference target $Y_{m,t}$ on the signals of all sources remains stationary over time. Further, the correlation structure among the source observations also remains stationary over time. This allows us to show that inference error is a time-invariant function of AoI, as we will see in Lemma \ref{lemma1}. It is practical to approximate time-varying functions as time-invariant functions in the scheduler design. Moreover, the scheduling policy developed for time-invariant AoI functions serves as a valuable foundation for studying time-varying AoI functions \cite{tripathi2021online}.

Now, we are ready to formulate our problem. We denote the scheduling policy as
$$\pi = (\pi_{m}(0), \pi_{m}(1), \ldots)_{m=1, 2, \ldots, M}.$$
We let $\Pi$ denote the set of all signal-agnostic and causal scheduling policies $\pi$ that satisfy (i) the scheduler makes its decision at every time $t$ based on the current and the past information available to the scheduler and (ii) the scheduler does not have access to the realization of the process $\{(Y_{m, t}, X_{m, t}), m=1, \ldots, M, t=0, 1, \ldots\}.$

Our goal is to find a policy $\pi \in \Pi$ that minimizes the discounted sum of inference errors: 
\begin{align}\label{Multi-scheduling_problem}
\mathrm{V}_{opt} =&\inf_{\pi \in \Pi}\mathbb{E}_{\pi} \left[\sum_{t=0}^{\infty}\gamma^t \sum_{m=1}^ML(Y_{m, t}, a_{m,t})\right], \\\label{Scheduling_constraint1}
&~\mathrm{s.t.} \sum_{m=1}^{M} \pi_{m}(t)=N, t=0, 1,\ldots, 
\end{align}
where $0 < \gamma <1$ is a discount factor, $L(Y_{m, t}, a_{m,t})$ is the inference error for the $m$-th target at time $t$, and $\mathrm{V}_{opt}$ is the minimum average inference error. {\blue For the simplicity of theoretical analysis, we consider that the loss function is normalized to $[0, 1]$. However, our results can easily be generalized to any bounded loss function.}

{\blue We also discuss the time-averaged version of the problem in Appendix \ref{avg}.}


\section{AoI-based Problem Reformulation} 

In this section, we first use an information-theoretic approach and AoI as tools to interpret how the correlation among the $m$-th target and the observation from $n$-th source affect the inference error for the $m$-th target. To that end, we use the concept of generalized conditional entropy \cite{dawid1998, farnia2016minimax} or specifically, the $L$-conditional entropy \cite{shishertimely}.

\subsection{An Information-theoretic Interpretation}\label{Interpretation}

For a random variable $Y$, the $L$-entropy is given by 
\begin{align} \label{gen_def_L_entropy}
H_L (Y) = \min _{a \in \mathcal{A}} \mathbb{E}_{Y \sim P_Y} [L (Y, a)].
\end{align}
The $L$-conditional entropy of $Y$ given $X = x$ is  \cite{dawid1998, farnia2016minimax, shishertimely}
\begin{align} \label{L_cond_en}
H_L (Y | X =x)
=\min _{a \in \mathcal{A}} \mathbb{E}_{Y \sim P_{Y|X=x}} [L (Y, a)]
\end{align}
and the $L$-conditional entropy of $Y$ given $X$ is 
\begin{align} 
\label{L_cond_en}
H_L (Y | X)
=\mathbb{E}_{X \sim P_{X}}[H_L (Y | X=x)].
\end{align} 
Moreover, an $L$-mutual information among two random variables $Y$ and $X$ is defined as \cite{dawid1998, farnia2016minimax, shishertimely}
\begin{align}
    I_L(Y;X)=H_L(Y)-H_L(Y|X).
\end{align}
The $L$-mutual information $I_L(Y;X)$ quantifies the reduction of expected loss in predicting $Y$ by observing $X$. The $L$-conditional mutual information among two random variables $Y$ and $X$ given $Z$ is defined as \cite{dawid1998, farnia2016minimax, shishertimely}
\begin{align}\label{CMI}
    I_L(Y;X|Z)=H_L(Y|Z)-H_L(Y|X, Z).
\end{align}

If AoI of sources are $(\Delta_1(t), \Delta_2(t), \ldots, \Delta_M(t))=(\delta_1, \delta_2, \ldots, \delta_M)$, then we can measure the reduction of expected inference error for target $Y_{m, t}$ given the freshest observations possible $X_{n, t-1}$ from source $n$ by using the definition of $L$-conditional mutual information \eqref{CMI}:
\begin{align}
    &I_L(Y_{m, t}; X_{n, t-1}|X_{1,t-\delta_1}, X_{2,t-\delta_2}, \ldots, X_{M,t-\delta_M})\nonumber\\
    &=H_L(Y_{m, t}|X_{1,t-\delta_1}, X_{2,t-\delta_2}, \ldots, X_{M,t-\delta_M})\nonumber\\
    &-H_L(Y_{m, t}|X_{1,t-\delta_1}, X_{2,t-\delta_2}, \ldots, X_{M,t-\delta_M},X_{n, t-1}).
\end{align}

Now, by using the concept of the $L$-conditional entropy, we equivalently express the multi-source scheduling problem in \eqref{Multi-scheduling_problem}-\eqref{Scheduling_constraint1} as the minimization of a penalty function of the AoI values $(\Delta_1(t), \Delta_2(t), \ldots, \Delta_M(t))$. 

Lemma \ref{lemma1} is first proved in \cite{shisher2021age}. We restate the result for the completeness of the paper.

\begin{lemma}\label{lemma1}
The inference error for $Y_{m, t}$ can be expressed as
\begin{align}
    \!\!\!\!\!\mathbb{E}_{\pi}\left[L(Y_{m, t}, a_{m,t})\right]\!\!=\!\!H_L\bigg(Y_{m, t}|(X_{n, t-\Delta_{n}(t)}, \Delta_{n}(t))_{n=1}^M\bigg)\!,
\end{align}
where $\Delta_{n}(t)$ is the AoI of source $n$ under the scheduling policy $\pi$. Moreover, if Assumption \ref{stationary} holds, then the $L$-conditional entropy
$$H_L\bigg(Y_{m, t}|(X_{n, t-\Delta_{n}(t)}, \Delta_{n}(t))_{n=1}^M\bigg)$$
is a function of AoI values $(\Delta_1(t), \Delta_2(t), \ldots, \Delta_M(t))$.
\end{lemma}
Lemma \ref{lemma1} implies that the inference error $\mathbb{E}_{\pi}\left[L(Y_{m, t}, a_{m,t})\right]$ can be represented as an L-conditional entropy of $Y_{m, t}$ given most recently delivered observations from all sources and their AoI values. Under Assumption \ref{stationary}, the L-conditional entropy is a function of AoI values $(\Delta_1(t), \Delta_2(t), \ldots, \Delta_M(t))$. For the simplicity of presentation, we represent the function as
\begin{align}\label{function}
&g_m(\Delta_1(t), \Delta_2(t), \ldots, \Delta_M(t)) \nonumber\\  
&=H_L\bigg(Y_{m, t}|(X_{n, t-\Delta_{n}(t)}, \Delta_{n}(t))_{n=1}^M\bigg).
\end{align}

By using Lemma \ref{lemma1} and \eqref{function}, we can express the problem \eqref{Multi-scheduling_problem}-\eqref{Scheduling_constraint1} as a minimization of the discounted sum of the AoI penalty functions.
\begin{align}\label{Multi-scheduling_problem1}
\mathrm{V_{opt}} =&\inf_{\pi \in \Pi}\mathbb{E}_{\pi} \left[\sum_{t=0}^{\infty}\gamma^t \sum_{m=1}^M g_m(\Delta_1(t), \Delta_2(t), \ldots, \Delta_M(t))\right], \\\label{Scheduling_constraint2}
&~\mathrm{s.t.} \sum_{m=1}^{M} \pi_{m}(t)=N, t=0, 1,\ldots, 
\end{align}

{However, as the number of sources $M$ and the number of channels $N$ increases, problem \eqref{Multi-scheduling_problem1}-\eqref{Scheduling_constraint2} becomes intractable. Specifically, problem \eqref{Multi-scheduling_problem1}-\eqref{Scheduling_constraint2} can be modeled as a Restless Multi-armed Bandit (RMAB) problem with a global penalty function:
$$\sum_{m=1}^M g_m\bigg(\Delta_1(t), \Delta_2(t), \ldots, \Delta_M(t)\bigg),$$
where AoI $\Delta_m(t)$ serves as the state of the $m$-th bandit arm. This problem is classified as ``restless" because even when a source $m$ is not scheduled for transmission, its AoI $\Delta_m(t)$ continues to evolve, incurring a penalty

$$g_m\bigg(\Delta_1(t), \Delta_2(t), \ldots, \Delta_M(t)\bigg).$$

In contrast to the traditional RMAB framework \cite{whittle1988restless}, the penalty function of each arm in this problem are intertwined. This is because the penalty function $g_m(\cdots)$ for each arm depends on the AoI values of all sources. This interconnectedness prevents the decomposition of the problem into independent Markov Decision Processes (MDPs), a common approach used for traditional RMABs \cite{whittle1988restless}, even with constraint relaxation. This inherent complexity distinguishes our problem from traditional RMAB problems and poses a greater analytical challenge.}

\section{Design of Scheduling Policy}
In this section, we will explore how to find a low complexity and close to optimal scheduling policy for \eqref{Multi-scheduling_problem1}-\eqref{Scheduling_constraint2}. We also consider that the scheduler knows the penalty function $g_m(\cdots)$ for any source $m$. In Section \ref{SchedulingUnknown}, we will design scheduling policy with unknown penalty function.

\subsection{Function Approximation $f_m(\delta_m)$}
{Like any MDP, our problem \eqref{Multi-scheduling_problem1}-\eqref{Scheduling_constraint2} can be solved by using dynamic programming method \cite{bertsekasdynamic1}. However, computing the optimal policy using dynamic programming becomes progressively harder in terms of space and time complexity for larger values of $M$, as the state-space and the action space to be considered grows exponentially with $M$. Whittle index policy \cite{whittle1988restless, shishertimely, Tripathi2019} and gain index policy \cite{shisher2023learning, OrneeMILCOM, gast2021lp} are low-complexity policies that have good performance for RMAB problems. For our problem, we can not directly obtain a scheduling policy like Whittle index policy and gain index policy. This is because the penalty function $g_m(\cdot, \cdot, \cdot)$ depends on the AoI of the other sources. To this end, in this section, we approximate the the penalty function $g_m(\cdot, \cdot, \cdot)$ by a function $f_m(\Delta_m(t))$ which only depends on the AoI of the source $m$. Then, we provide an index-based scheduling policy. 

We can define a possible option for $f_m(\delta_m)$ as follows:
\begin{align}\label{option1}
f_{m}(\delta_m)= H_L(Y_{m,t}|X_{m, t-\delta_m},Z_{-m, t-1})
\end{align}
where $Y_{m,t}$ is the $m$-th target at time $t$, $X_{m, t-\delta_m}$ is the observation generated from source $m$ at time $t - \delta_m$, $Z_{-m, t-1}=[X_{n, t-1}: \forall n\neq m]$ is a vector containing the signal observations generated at time $t-1$ from all sources except source $m$. This definition of $f_m(\delta_m)$ captures the expected inference error of the current target $Y_{m,t}$ given observation $X_{m, t-\delta_m}$ of source $m$ generated $\delta_m$ time slots ago and the recent observations $Z_{-m, t-1}$ from other sources generated $1$ time slot ago.

Since we have knowledge of the penalty function $g_m(\delta_1, \delta_2, \ldots, \delta_M)$, we can evaluate the function $f_{m}(\delta_m)$, by setting $\delta_n=1$ for all $n \neq m$. This effectively evaluates the function under the scenario where all other sources have an AoI of $1$ except source $m$. In Appendix \ref{approximatefunction}, we discuss how to compute the function $f_m(\delta)$ without using $g_m(\cdot, \cdot,\cdot)$.}

\begin{lemma}\label{approx}
   If the sequence $(Y_{m, t}, X_{n, t-1}, X_{n, t-k})$ forms a Markov chain $Y_{m, t} \leftrightarrow X_{n, t-1} \leftrightarrow X_{n, t-k}$ for all $k\geq 1$, then the following assertions are true.

   (a) We have 
   \begin{align}
         g_m(\delta_1, \delta_2, \ldots, \delta_M) \geq f_m(\delta_m).
    \end{align}
    
   (b) We have 
    \begin{align}\label{lowerbound}
        \!\!\!g_m(\delta_1, \delta_2, \ldots, \delta_M)\!\!=\!f_m(\delta_m)+O(\max_{n\neq m}\epsilon_{m,n}^2).
    \end{align}
     where the parameter $\epsilon_{m,n}$ is determined as 
    \begin{align}\label{condition1}
     &\epsilon_{m,n}\nonumber\\
     &= \sqrt{\max_{\delta_1, \delta_2, \ldots, \delta_M} I_L(Y_{m, t}; X_{n, t-1}|X_{1, t-\delta_1}, \ldots, X_{M, t-\delta_M})}, 
    \end{align}
    and $L$-conditional mutual information $I_L( )$ is defined in \eqref{CMI}.
       \end{lemma}

\begin{proof}
    See Appendix \ref{papprox}.
\end{proof}

By the definition of Markov chain $Y_{m, t} \leftrightarrow X_{n, t-1} \leftrightarrow X_{n, t-k}$, the target $Y_{m, t}$ does not depend on the older observation $X_{n, t-k}$ given a new observation $X_{n, t-1}$ from source $n$. Lemma \ref{approx}(a) implies that If the Markov chain $Y_{m, t} \leftrightarrow X_{n, t-1} \leftrightarrow X_{n, t-k}$ holds, then $f_m(\delta_m)$ is a lower bound of $g_m(\delta_1, \delta_2, \ldots, \delta_M)$. 

In \eqref{condition1}, we measure the maximum reduction of expected inference error for target $m$ after adding a new observation from source $n$ by the parameter $\epsilon_{m, n}$. If the parameter $\epsilon_{m, n}$ is close to zero, the impact of correlation between target $Y_{m,n}$ and a new observation from source $n$ on inference error of $Y_{m,n}$ is nearly negligible. We say if $\epsilon_{m, n}$ is close to zero, correlation is low. 

Lemma \eqref{approx}(b) implies that if $\epsilon_{m,n}$ tends to zero for all $n\neq m$, then the approximation error goes to zero, i.e., the function $f(\delta_m)$ is a good approximation of inference error function $g(\delta_1, \delta_2, \ldots, \delta_M)$ under low correlation.



\begin{remark}
    Our information-theoretic lower bound, $f_m(\delta_m)$, extends the bound in \cite{vishrantcorrelated} to a significantly broader class of systems, including non-linear and non-Gaussian dynamics. As shown in Appendix \ref{approximatefunction}, this bound coincides with the bound in \cite[Theorem 2]{vishrantcorrelated} when specialized to LTI systems with Gaussian processes.
\end{remark}

\subsection{Scheduling Policy}
{\blue To design a scheduling policy, we first divide entire duration into episodes, where each episode $k$ contains $T$ time slots. Here, we approximate infinite time with finite time $T$. For discounted penalty problem, the approximation error in considering finite time $T$ is negligible if $T$ is a very large number. Next, we use approximated penalty function. Thus, we reformulated the problem \eqref{Multi-scheduling_problem1}-\eqref{Scheduling_constraint2} as follows: 
\begin{align}\label{Multi-scheduling_problem2}
\mathrm{V}_{f, opt}(T) =&\inf_{\pi \in \Pi}\mathbb{E}_{\pi} \left[\sum_{t=0}^{T-1}\gamma^t \sum_{m=1}^M f_m(\Delta_m(t))\right], \\\label{Scheduling_constraint3}
&~\mathrm{s.t.} \sum_{m=1}^{M} \pi_{m}(t)=N, t=0, 1,\ldots, 
\end{align}
where $T$ is the truncated time horizon, $\mathrm{V}_{f, opt}(T)$ is the optimal objective value of the approximated problem \eqref{Multi-scheduling_problem2}-\eqref{Scheduling_constraint3} and $\pi_{f, opt}$ be an optimal policy to \eqref{Multi-scheduling_problem2}-\eqref{Scheduling_constraint3}.}


\begin{theorem}\label{approximationBound}
{\blue We have 
\begin{align}\label{theo1}
    |\mathrm{V}_{{\pi_{f, opt}}}(T)-\mathrm{V}_{opt}(T)|\leq \frac{2(1-\gamma^T)}{1-\gamma}\sum_{m=1}^MO\bigg(\max_{n\neq m}\epsilon_{m, n}^2\bigg),
\end{align}
where $\mathrm{V}_{{\pi_{f, opt}}}(T)$ is the optimal objective value of \eqref{Multi-scheduling_problem2}-\eqref{Scheduling_constraint3} and $\mathrm{V}_{opt}(T)$ is the optimal objective value of the problem \eqref{Multi-scheduling_problem1}-\eqref{Scheduling_constraint2} with finite horizon $T$.}
\end{theorem}
\begin{proof}
    See Appendix \ref{papproximationBound}.
\end{proof}
Theorem \ref{approximationBound} provides an approximation bound. Theorem \ref{approximationBound} implies that if the correlation among the sources is low, i.e., $\epsilon_{m, n}$ is close to $0$ for all $n\neq m$, then the optimal policy $\pi_{f, opt}$ to the approximated problem \eqref{Multi-scheduling_problem2}-\eqref{Scheduling_constraint3} will yield close to optimal performance for the main problem \eqref{Multi-scheduling_problem1}-\eqref{Scheduling_constraint2}.

\subsubsection{Lagrangian Relaxation}
{After establishing the approximation bound in Theorem \ref{approximationBound}, our goal is to find scheduling policy for the approximated problem \eqref{Multi-scheduling_problem2}-\eqref{Scheduling_constraint3}. Due to constraint \eqref{Scheduling_constraint3}, finding an optimal policy for the approximated problem \eqref{Multi-scheduling_problem2}-\eqref{Scheduling_constraint3} remains computationally intractable. This problem has been proven to be PSPACE-hard \cite{papadimitriou1994complexity}. We relax the constraint and get the relaxed problem:
{\blue \begin{align}\label{relaxedproblem}
&\mathrm{V}_{f, opt}(\boldsymbol{\lambda}, T)\nonumber\\
&=\inf_{\pi \in \Pi}\mathbb{E}_{\pi} \left[\sum_{t=0}^{T-1}\gamma^t \bigg(\sum_{m=1}^M f_m(\Delta_m(t))+\lambda_t (\pi_m(t)-N)\bigg)\right],
\end{align}
where $\lambda_t \in \mathbb R$ is the Lagrangian multiplier, $\boldsymbol{\lambda}=(\lambda_1, \lambda_2, \ldots, \lambda_T)$, and $\mathrm{V}_{f, opt}(\boldsymbol{\lambda}, T)$ is the optimal objective value of the relaxed problem \eqref{relaxedproblem}.} 

The dual problem to \eqref{relaxedproblem} is given by
\begin{align}\label{dualP}
    \max_{\boldsymbol{\lambda}\in \mathbb R^T} \mathrm{V}_{f, opt, \lambda},
\end{align}
where $\boldsymbol{\lambda}^*$ denotes the optimal Lagrange value to the problem \eqref{dualP}.
} 

\subsubsection{Gain Index and Threshold Policy}
The problem \eqref{relaxedproblem} can be decomposed into $M$ sub-problems, where each sub-problem $m$ is given by

\begin{align}\label{sub-relaxedproblem}
\inf_{\pi_m \in \Pi_m}\mathbb{E}_{\pi} \left[\sum_{t=0}^{\infty}\gamma^t \bigg(f_m(\Delta_m(t))+\lambda_t \pi_m(t)\bigg)\right].
\end{align}
{\blue The optimal value function for each source $m$ at time $t$ is given by
\begin{align}\label{decomposedValue}
    &J_{m,\boldsymbol{\lambda}, t}(\delta)\nonumber\\
    &=f_m(\delta)+\min\bigg\{\lambda_t+\gamma J_{m, \boldsymbol{\lambda}, t+1}(1), \gamma J_{m,\boldsymbol{\lambda}, t+1}(\delta+1)\bigg\},
\end{align}
where $\delta$ is the current AoI value and $J_{m,\boldsymbol{\lambda}, t}$ is the value function associated with the problem \eqref{sub-relaxedproblem}. The value function can be obtained by using backward induction method \cite{bertsekasdynamic1}.

\begin{lemma}\label{thresholdPolicy}
    Given the AoI value $\Delta_m(t)=\delta$, the scheduling decision $\pi_m(t)=1$ is optimal to \eqref{sub-relaxedproblem}, if the following holds: 
\begin{align}\label{thresholdpolicy}
J_{m,\boldsymbol{\lambda}, t+1}(\delta+1)-J_{m,\boldsymbol{\lambda}, t+1}(1)> \frac{\lambda_t}{\gamma}.
\end{align}
\end{lemma}}
\begin{proof}
    Lemma \ref{thresholdPolicy} holds due to \eqref{decomposedValue}.
\end{proof}

{\blue Lemma \ref{thresholdPolicy} implies that if the difference between $$J_{m,\boldsymbol{\lambda}, t+1}(\delta+1)-J_{m,\boldsymbol{\lambda}, t+1}(1)$$ exceeds a threshold value $\frac{\lambda_t}{\gamma}$, then scheduling source $m$ at time $t$ is optimal to the problem \eqref{sub-relaxedproblem}.} However, the policy can not be applied to the main problem due to scheduling constraint. Towards this end, following \cite{shisher2025computation, OrneeMILCOM, shisher2023learning}, we define gain index. 

{\blue\begin{definition}[\textbf{Gain Index}]\label{gainindex}
 For AoI value $\delta$ and Lagrangian multiplier $\boldsymbol{\lambda}$, the gain index is given by
\begin{align}\label{gainindex}
\alpha_{m, \boldsymbol{\lambda}, t}(\delta):=J_{m,\boldsymbol{\lambda}, t+1}(\delta+1)-J_{m,\boldsymbol{\lambda}, t+1}(1).
\end{align}
\end{definition}}

The gain index $\alpha_{m, \boldsymbol{\lambda}, t}(\delta)$ quantifies the total reduction in the discounted expected sum of inference errors from time $t+1$ to time $T-1$, when action $\pi_{m}(t)=1$ is chosen at time $t$ over $\pi_{m}(t)=0$, where the latter implies $m$-th source is not scheduled at time $t$. This metric enables strategic resource utilization at each time slot to enhance overall system performance. {\blue Prior works \cite{shisher2023learning, OrneeMILCOM, shisher2025computation} used action-value function to define the gain index, whereas \cite{brown2020index} called it Lagrangian index.} In \cite{shisher2023learning, OrneeMILCOM, gast2021lp, brown2020index}, source $m$ is scheduled if the difference between two action value $Q_{m, \boldsymbol{\lambda}, t}(\delta, 0)-Q_{m, \boldsymbol{\lambda}, t}(\delta, 1)$ becomes greater than $0$, where 
\begin{align}
    Q_{m, \boldsymbol{\lambda}, t}(\delta, 0)&=f_m(\delta)+ \gamma J_{m,\boldsymbol{\lambda}, t}(\delta+1)\\
    Q_{m, \boldsymbol{\lambda}, t}(\delta, 1)&=f_m(\delta)+\lambda_t+\gamma J_{m, \boldsymbol{\lambda}, t}(1).
\end{align}
Because the instantaneous penalty $f_m(\delta)$ incurred at time $t$ is same for whichever decision we take, the policies derived from ($Q_{m, \boldsymbol{\lambda}, t}(\delta, 0)-Q_{m, \boldsymbol{\lambda}, t}(\delta, 1)$) and our gain index-based single source scheduling are equivalent.

For the main problem, our algorithm is provided in Algorithm \ref{alg:MGF}. We call it Maximum Gain First (MGF) Policy. At each time $t$, we schedule $N$ sources with highest gain index $\alpha_{m, \boldsymbol{\lambda}_{t,k}, t}(\Delta_m(t))$. {\blue If multiple sources have same gain index, we use a tie-breaker index function $\psi_{t, m}(\delta)$ that can rank sources with same gain index value at any time $t$. We can use simple random mixing or Markov random mixing policy to design a tie-breaker index function as discussed in \cite{brown2020index}.}

After every episode $k$, we update Lagrange multiplier by using sub-gradient ascent method as follows:
\begin{align}\label{lagrangeUpdate}
    \lambda_{t, k+1}=&\lambda_{t,k}\nonumber\\
    +&\frac{\theta}{k}\bigg(\gamma^t\sum_{m=1}^M \mathbf 1\bigg(\alpha_{m, \boldsymbol{\lambda}_{t,k}, t}(\Delta_m(t))>\frac{\lambda_{t,k}}{\gamma}\bigg)-\frac{N}{1-\gamma}\bigg),
\end{align}
where $\mathbf 1(\cdot)$ is an indicator function and $\theta/k>0$ is step size.

\subsection{Performance Analysis}
{\blue In this section, we analyze the performance of Algorithm \ref{alg:MGF} with an optimal policy.} 

{\blue Let $\pi_{\text{gain}}$ denote the policy presented in Algorithm \ref{alg:MGF} and $\mathrm{V}_{\mathrm{gain}}(\boldsymbol{\lambda}, T)$ represent the objective value under policy $\pi_{\text{gain}}$ associated with Lagrange multiplier $\boldsymbol{\lambda}$. 

\begin{theorem}\label{theorem2}
     Given a discount factor $\gamma\in(0.5, 1)$, an optimal tie-breaker function $\psi$, and the fixed ratio $r=N/M$, if $\boldsymbol{\lambda_k}=\boldsymbol{\lambda^*}$, then in episode $k$, we have  
    \begin{align}\label{theo2}
   &\frac{|\mathrm{V}_{\mathrm{gain}}(\boldsymbol{\lambda^*, T})-\mathrm{V}_{opt}(T)|}{M}\nonumber\\
   &\leq \frac{2(1-\gamma^T)}{(1-\gamma)}\frac{\sum_{m=1}^{M}O\bigg(\max_{n\neq m}\epsilon_{m, n}^2\bigg)}{M}+O(\beta\sqrt \frac{(1-r)r}{M}),
\end{align}
where $$\beta=\frac{2(2\gamma)^{T}}{2\gamma-1}.$$
\end{theorem}
\begin{proof}
    See Appendix \ref{ptheorem2}.
\end{proof}}
{\blue Theorem \ref{theorem2} shows the normalized optimality gap of our policy depends on the correlation parameter $\epsilon_{m, n}$ and the system size $M$. In general, for large system size i.e., $M\to \infty$, the optimality gap can be finite, specifically we have 
\begin{align}
    &\lim_{M\to \infty} \frac{|\mathrm{V}_{\mathrm{gain}}(\boldsymbol{\lambda^*, T})-\mathrm{V}_{opt}(T)|}{M}\nonumber\\
    &\leq \lim_{M\to\infty} \frac{2(1-\gamma^T)}{(1-\gamma)}\frac{\sum_{m=1}^{M}O\bigg(\max_{n\neq m}\epsilon_{m, n}^2\bigg)}{M}.
\end{align}
However, for a special case, if the sum of the correlation parameters satisfies 
\begin{align}\label{correlationCondition}
    \sum_{m=1}^{M}O\bigg(\max_{n\neq m}\epsilon_{m, n}^2\bigg)=o(M),
\end{align} 
we have 
\begin{align}
    \lim_{M\to\infty} \frac{2(1-\gamma^T)}{(1-\gamma)}\frac{\sum_{m=1}^{M}O\bigg(\max_{n\neq m}\epsilon_{m, n}^2\bigg)}{M}=0.
\end{align}
For distributed sensor networks with a large number of sensors, where only a few sensors exhibit mutual correlation, the special case \eqref{correlationCondition} can hold. 

\subsection{Algorithm Simplification}
In Algorithm \ref{alg:MGF}, we consider a large sequence of Lagrange multiplier $\boldsymbol{\lambda}=(\lambda_0, \lambda_2, \ldots, \lambda_{T-1})$. This leads to a Lagrangian dual problem \eqref{dualP} that is practically difficult to solve in optimality. To address this challenge, we can simplify Algorithm \ref{alg:MGF} by considering $\lambda_t=\lambda$ for all $t$. For this modification, our Lagrange update rule becomes
\begin{align}
    \lambda_{k+1}=&\lambda_k+\frac{\theta}{k}\bigg(\sum_{t=0}^{T-1}\gamma^t\sum_{m=1}^M \mathbf 1\bigg(\alpha_{m, \lambda_k}(\Delta_m(t))>\frac{\lambda_k}{\gamma}\bigg)\nonumber\\
    &-\frac{(1-\gamma^T)N}{1-\gamma}\bigg).
\end{align}}

\begin{algorithm}[t]
\SetAlgoLined
\SetAlgoNoEnd
\SetKwInOut{Input}{input}
\caption{Maximum Gain First Policy}\label{alg:MGF}
\emph{Input $\boldsymbol{\lambda_k}=0$}\\
\For{episode $k=1, 2 \ldots$}
{Evaluate $\alpha_{m, \boldsymbol{\lambda_{t,k}, t}}(\delta)$ for all $m$ and $\delta$\\
{\blue Update $\mathrm{Subgrad}(t) \gets 0$}\\
\For{every time $t$ in episode $k$}
{\emph{Update $\Delta_{m}(t)$ for all $m$}\\
\emph{Initialize $\pi_m(t)\gets 0$ for all $m$}\\
{\blue $\alpha_m \gets \alpha_{m, \boldsymbol{\lambda_{t,k}, t}}(\Delta_m(t))$}\\
Schedule $N$ sources with highest $\alpha_m$\\
{\blue Sources with same gain index are ranked\\
with a tie-breaker index function $\psi_{t, m}$} \\
{\blue $\mathrm{Subgrad}(t)\gets \sum_{m=1}^M\gamma^{t}\mathbf 1(\alpha_m>\frac{\lambda_{t,k}}{\gamma})-N$}
}\emph{Update $\boldsymbol{\lambda_{k+1}}$ using \eqref{lagrangeUpdate}}}
\end{algorithm}

\section{Design of Online Scheduling Policy with Unknown Penalty Function}\label{SchedulingUnknown}
%

As discussed in the previous section, Algorithm \ref{alg:MGF} can be applied when (i) the penalty function $g_m(\delta_1, \delta_2, \ldots, \delta_M)$ or its approximation $f_m(\delta_m)$ for all $m$ are known and (ii) the correlation $\epsilon_{m, n}$ among the sources are low. In this section, we establish an online scheduling policy that does not know the penalty function $g_m(\delta_1, \delta_2, \ldots, \delta_M)$, its approximation $f_m(\delta_m)$, and the distribution of the targets and the source observations. Moreover, we consider an arbitrary correlation structure, i.e., not limited to low correlation like previous section.  

If the penalty function or the observation sequences $(Y_{m,t}, (X_{n, t-\mu_n})_{n=1}^M)$ are not known to the scheduler, it is not possible to get an online estimate of the function $f_m$ defined in \eqref{option1}. This is because it requires AoI values of all $M-1$ sources equal to $1$ at some time slots $t$. But, it is impossible to schedule $M-1$ sources at any time $t$ unless $N=M-1$.
For the online setting, we consider a new function approximation: 
\begin{align}
    f_{m, k}(\delta)= \mathbb E_{\pi_k}[L(Y_{m,t}, a_{m,t})|\Delta_m(t)=\delta],
\end{align}
which is the expected inference error of target $m$ under policy $\pi_k$ at episode $k$ when the AoI of source $m$ is $\Delta_m(t)=\delta$. The approximated function $f_{m, k}(\delta)$ depends on the policy $\pi_k$ of episode $k$. Because the function $f_{m, k}(\delta)$ depends on the policy, it is not possible to get the function $f_{m, k}(\delta)$ before the episode $k$. Hence, the scheduler can only decide policy $\pi_k$ based on the estimation of the functions $f_{m, k-1}, f_{m, {k-2}}, \ldots, f_{m, {0}}$.

For designing online scheduling, we consider the following setting: 
\begin{itemize}
    \item The entire time horizon is divided into episodes $$k = 0, 1, 2, \ldots,$$ each consisting of $T$ time slots. The value of $T$ is chosen to be sufficiently large such that approximating the infinite horizon discounted penalty problem with a finite horizon of length $T$ introduces negligible error.
    \item {\blue Computing the loss function for all time $t$ and in real-time is computationally infeasible because ground truth target is not immediately available. To address this challenge, we adopt a bandit feedback structure. This structure allows the receiver to compute the loss function only for the scheduled time slots. The receiver stores the following quantities for every source $m$: the predicted output $a_{m, t}$, the most recent observation $X_{m, t-\Delta_m(t)}$, and $\pi_m(t)$ or equivalently $\Delta_m(t)$. Let's consider $\pi_m(j)=1$, i.e., source $m$ is scheduled at time $t=j$. In this case, the receiver obtains $X_{m, j}$ at time $t=j+1$. Since we have assumed the true target $Y_{m, t}$ can be inferred directly from $X_{m, t}$ (e.g., $Y_{m, t}$ can be a function of, or equivalent to, $X_{m, t}$), the receiver can use the stored $X_{m, j}$ to compute $Y_{m, j}$. With both the prediction $a_{m, j}$ and the computed true label $Y_{m, j}$, the receiver can exactly compute the loss $L(Y_{m, j}, a_{m, j})$. Crucially, this loss is only available at the scheduled time slots $j$ and not for all $t$, which defines the partial feedback characteristic of the bandit setting.}

    \item At the beginning of every episode $k$, the scheduler decides a policy $\pi_k \in \Pi$ by using an estimate of the function $f_{m, {k-1}}(\Delta_m(t))$.
\end{itemize}

\subsection{Online Scheduling for Lagrangian Relaxed Problem}
At first, we establish an online scheduling policy for Lagrangian relaxed problem. In each episode $k$, we replace $f_m(\Delta_m(t))$ in \eqref{sub-relaxedproblem} by the function $f_{m, {k-1}}(\Delta_m(t))$ and solve the following Lagrangian relaxed problem for each source $m$:
\begin{align}\label{sub-relaxedproblemonline}
\inf_{\pi_m \in \Pi_m}\mathbb{E}_{\pi} \left[\sum_{t=0}^{\infty}\gamma^t \bigg(f_{m,{k-1}}(\Delta_m(t))+\lambda \pi_m(t)\bigg)\right].
\end{align}


Since the true function $f_{m, {k-1}}(\delta)$ is unknown, we use an empirical estimate, $\hat f_{m, k-1}(\delta)$, calculated as the average loss for a given state $\delta$ during episode $k-1$:
\begin{align}\label{empiricalestimate}
    \hat f_{m, k-1}(\delta)=\frac{\sum_{t=T_{k-1}}^{T_k-1}\mathbf 1(\Delta_{m}(t)=\delta, \pi_{m}(t)=1) L(Y_{m,t}, a_{m,t})}{\mathcal{N}_{m,k}(\delta)},
\end{align}
where $T_k$ is the starting time of episode $k$ and the frequency $\mathcal{N}_{m,k}(\delta)$ is 
\begin{align}
    \mathcal{N}_{m,k}(\delta)=\max\left(\sum_{t=T_{k-1}}^{T_k-1}\mathbf 1(\Delta_{m}(t)=\delta, \pi_{m}(t)=1), 1\right).
\end{align}

{\blue However, the estimate $\hat f_{m, k-1}(\delta)$ will naturally deviate from the true value. Given a chosen exploration parameter $\eta\in(0, 1)$, we account for the uncertainty by defining  the following confidence radius
\begin{align}
d_{m, k}(\delta)=\sqrt{\frac{{\mathrm{ln} (\frac{2}{\eta}})}{2\mathcal{N}_{m,k}(\delta)}},
\end{align}
which shrinks as we gather more samples for that state. Moreover, the parameter $\eta$ also controls the confidence radius: The confidence radius increases as the parameter $\eta$ decreases.  

The confidence radius creates a ball of plausible values, $B_{m,k}(\delta)$: 
\begin{align}
    &B_{m,k}(\delta)\nonumber\\
    &=\{f_{m,k}(\delta): |f_{m,k}(\delta)-\hat f_{m, k}(\delta)|\leq d_{m, k}(\delta)\},
\end{align}
Lemma \ref{prob} shows that the actual value $f_{m,k}(\delta)$ is in the ball of plausible values, $B_{m,k}(\delta)$ with probability $1-\eta$. As the parameter $\eta$ decreases, the ball $B_{m,k}(\delta)$ expands and the uncertainty reduces. However, as we will see next that the optimistic estimate of the actual value deviates away from the empirical estimate as $\eta$ decreases.}

Following the principle of ``optimism in the face of uncertainty" \cite{auer2006logarithmic, auer2008near, auer2002finite}, we select an optimistic estimate, $\tilde f_{m, {k-1}}(\delta)$, given by
\begin{align}\label{estimatef}
   \tilde f_{m, {k-1}}(\delta)=\max\{\hat f_{m, k-1}(\delta)-d_{m, k}(\delta), 0\}.
\end{align}
{\blue As the parameter $\eta$ decreases, the confidence radius $d_{m, k}(\delta)$ increases and the estimate $\tilde f_{m, {k-1}}(\delta)$ deviates away from empirical value.} 

This optimistic value $\tilde f_{m, {k-1}}(\delta)$ is used in a value iteration algorithm to update our policy, ensuring we explore efficiently while accounting for statistical uncertainty: 
\begin{align}\label{estimateV}
    &J_{m, \mathrm{optimistic}, k}(\delta)=\tilde f_{m, {k-1}}(\delta)\nonumber\\
    &+\min\bigg\{\lambda+\gamma J_{m, \mathrm{optimistic}, k}(1), \gamma J_{m, \mathrm{optimistic}, k}(\delta+1)\bigg\},
\end{align}
where $J^{(m)}(\delta)$ can obtained by using value iteration algorithm \cite{bertsekasdynamic1}. 

{\blue To ensure convergence of the estimated value function, we introduce a weighting parameter $\zeta \in (0, 1)$. This parameter controls the trade-off between relying on the optimistic value function and retaining information from the history. Using $\zeta$, the estimated value function is defined as follows for all $\delta$:}
\begin{align}\label{estimatevalue}
    \tilde J_{m, \lambda, k}(\delta)= \zeta^k J_{m,\mathrm{optimistic}, k}(\delta)+(1-\zeta^k)\tilde J_{m, \lambda, k-1}(\delta),
\end{align}
where $0<\zeta<1$. This step is necessary to ensure convergence to a stable solution.


\begin{algorithm}[t]
\SetAlgoLined
\SetAlgoNoEnd
\SetKwInOut{Input}{input}
\caption{Online Policy for \eqref{sub-relaxedproblemonline}}\label{alg:threshold}
\For{episode $k=1, 2 \ldots$}
{\emph{Update} $\tilde f_{m, k-1}$ using \eqref{estimatef}.\\
{Evaluate $\tilde J_{m,\lambda,k}$ for all $m$ and $\delta$} using \eqref{estimatevalue}. \\
\For{every $t$ in episode $k$}
{\emph{Update $\Delta_{m}(t)$ for all $m$}\\
Initialize $\pi_m(t)\gets 0$ for all $m$\\
\For {$m=1, 2 \ldots, M$}
{\If{$\tilde J_{m,\lambda,k}(\Delta_{m}(t)+1)-\tilde J_{m,\lambda,k}(1) > \frac{\lambda}{\gamma}$}
{$\pi_m(t)\gets 1$; \emph{Schedule Source $m$}}}
}}
\end{algorithm}

\begin{algorithm}[t]
\SetAlgoLined
\SetAlgoNoEnd
\SetKwInOut{Input}{input}
\caption{Online Maximum Gain First Policy}\label{alg:gain}
\emph{Input: Lagrange Multiplier $\lambda_1=0$} \\
\For{episode $k=1, 2 \ldots$}
{\emph{Update $\tilde f_{m, k-1}$ by using \eqref{estimatef}}\\
{Evaluate $\tilde J_{m,\lambda,k}$ for all $m$ and $\delta$} by using \eqref{estimatevalue}\\
\For{every $t$ in episode $k$}
{\emph{Update $\Delta_{m}(t)$ for all $m$}\\
{Compute Gain index $\alpha_m:=\tilde J_{m,\lambda,k}(\Delta_{m}(t)+1)-\tilde J_{m,\lambda,k}(1)$}\\
{Schedule $N$ sources with highest $\alpha_m$}
\\
$\mathrm{Subgrad}\gets \mathrm{Subgrad}+\sum_{m=1}^M\gamma^{t}\mathbf 1(\alpha_m>\frac{\lambda}{\gamma})$
}}\emph{Update $\lambda_{k+1}\gets \lambda_k+(\theta/k)(\mathrm{Subgrad}-\frac{N}{1-\gamma})$}
\end{algorithm}

  

Following \eqref{estimatevalue}, policy $\pi_k$ is provided in Algorithm \ref{alg:threshold}. The actual value function of policy $\pi_k$ associated with the problem \eqref{sub-relaxedproblemonline} can be written as 

\begin{align}\label{actualvalue}
    &J_{m,\lambda, k}(\delta)\nonumber\\
    &=f_{m, {k}}(\delta)+\min\bigg\{\lambda+\gamma J_{m, \lambda, k}(1), \gamma J_{m,\lambda, k}(\delta+1)\bigg\}.
\end{align}


Let $J^*$ be the optimal value function. Now, we are ready to establish the following results. 
\begin{theorem}\label{convergence}
    The following assertions are true.
    
   (a) There exists a value $J_{m} \in \mathbb R$ such that $$\lim_{k\to \infty}\tilde J_{m, \lambda_k, k}(\delta)=J_{m}(\delta).$$
   
   (b) {\blue Given an exploration parameter $\eta\in (0, 1)$, for any episode $k$ with probability $1-\eta$, we have} 
        \begin{align}
            &\bigg|J_{m,\lambda, {k}}(\delta)-J_{m,\lambda, {k-1}}(\delta)\bigg|\nonumber\\
            &\leq \frac{1}{1-\gamma}\bigg(2\max_{\delta} d_{m, k-1}(\delta)+\beta_k\bigg),
        \end{align}
        where 
        \begin{align}\label{difference}
           \beta_k=\max_{\delta} |f_{m, {k-1}}(\delta)-f_{m, {k}}(\delta)|
        \end{align}

\end{theorem}
\begin{proof}
See Appendix \ref{convergenceproof}.
\end{proof}
Theorem \ref{convergence}(a) implies that the estimated value function $\tilde J_{m, \lambda_k, k}$ converges to a stable value function. As the policy $\pi_k$ depends on $\tilde J_{m, \lambda_k, k}$, we can say that the policy also converges to a stable solution, but it is not guaranteed to converge to an optimal policy.

Theorem \ref{convergence}(b) characterizes the distance between the value function of $k$-th episode and the value function of the $k-1$-th episode. According to the theorem, the distance is upper bounded by a confidence radius $d_{m, k-1}$ and $\beta_k$. The term $d_{m, k-1}$ is related to the uncertainty of the estimation of $f_{m, k-1}$. If number of time slot $T$ in every episode increases, the confidence radius $d_{m, k-1}$ decreases. On the other hand, we show in Fig. \ref{fig:beta} that $\beta_k$ decreases to $0$ for smaller $\zeta$.


\subsection{Online Maximum Gain First Policy}
Algorithm \ref{alg:threshold} can not be applied as it does not satisfy the constraint \eqref{Scheduling_constraint2}. For this reason, similar to Algorithm \ref{alg:MGF}, we utilize the concept of gain index defined in Definition \eqref{gainindex} and the structure of the Online Threshold Policy provided in Algorithm \ref{alg:threshold} and develop a new ``Online Maximum Gain First" policy in Algorithm \ref{alg:gain}.

Initially, we take Lagrange multiplier $\lambda_1=0$ as input. Then, the Online Maximum Gain First policy proceeds as follows:
\begin{itemize}
    \item At the beginning of every episode $k$, we update function $\tilde f_{m, k-1}$ by using \eqref{estimatef}, where we use the feedback of inference loss after episode $k-1$. 
    \item After that by using the function $\tilde f_{m, k-1}$, we compute value function $\tilde J_{m, \lambda, k}$ by using \eqref{estimatevalue}. 
    \item At every time in episode $k$, we update AoI value $\Delta_m(t)$ for all $m$ and compute the gain index 
    \begin{align}
        \alpha_m=\tilde J_{m, \lambda, k}(\Delta_m(t)+1)-\tilde J_{m, \lambda, k}(1).
    \end{align}
    \item $N$ sources with highest gain index are scheduled at every time $t$.
    \item After every episode, we update $\lambda_{k+1}$ using sub-gradient ascent method. 
\end{itemize}
At the initial episode $k=0$, the scheduler can apply Maximum Age First policy.

\section{Simulation Results}\label{simulation}
{\blue Though our theoretical analysis consider finite space $\mathcal X_m$ and $\mathcal Y_m$, our algorithms are general and can be applied to continuous space. Our algorithms only require countable AoI values. In this section, we will illustrate the performance of our Algorithm \ref{alg:MGF} and Algorithm \ref{alg:gain} on a continuous space $\mathcal X_m$ and $\mathcal Y_m$ setting.}   

\subsection{Simulation Model}

We use a model from \cite{tripathi2022optimizing} for designing the correlation structure between sources. 
At the beginning of every time-slot, each source $m$ collects
information about its own state. In addition, with probability $p_{m, n}$, the update collected by source $m$ also contains information about the current state of source $n$. Motivating examples of such correlated sources are cameras with overlapping fields of view and sensors with spatial correlation between the processes being monitored. Here are two examples:

{\blue 
{\bf Cameras with Overlapping Fields of View:} If camera $m$ successfully updates, there is a probability ($p_{m,n}$) that it also provides an updated view of targets associated with an adjacent camera $n$. 

{\bf Spatially Correlated Environmental Sensors:} Sensors monitoring processes (e.g., temperature, air quality, fire, rain) often exhibit spatial correlation. The probability ($p_{m,n}$) that one sensor's update contains useful, non-redundant information about an adjacent sensor's state. }

We denote the state of source $m$ by $Z_{m, t}$. The state $Z_{m, t}$ is considered to evolve as
\begin{align}\label{simsystem}
    Z_{m, t}= a_m Z_{m, t-1}+W_{m, t},
\end{align}
where $W_{m, t}$ follows zero mean Gaussian distribution with variance $1$ and $W_{m, t}$ are i.i.d. over time $t$ and independent over source $m$.  Let $X_{m, t}$ be the update from source $m$. The update $X_{n,t}$ includes state $Z_{n, t}$ of the source $n$ and $Z_{m, t}$ of source $n$ with probability $p_{m, n}$. A value of $p_{m, n}=0$ suggests that there is never any information at source $n$ about the state of source $m$, while a value of $p_{m, n}=1$  suggests that source $n$ has complete information about state of source $m$ at all times. In all of our simulation, we consider source $1$ gets information about states of other sources with probability $p_{m,1}=p$. Sources except $1$ never get any information of other sources, i.e., $p_{m, n}=0$ for all $n\neq 1$ and $m\neq n$. In this simulation, we have used discount factor $\gamma=0.7$.

\begin{figure}[t]
    \centering
\includegraphics[width=0.35\textwidth]{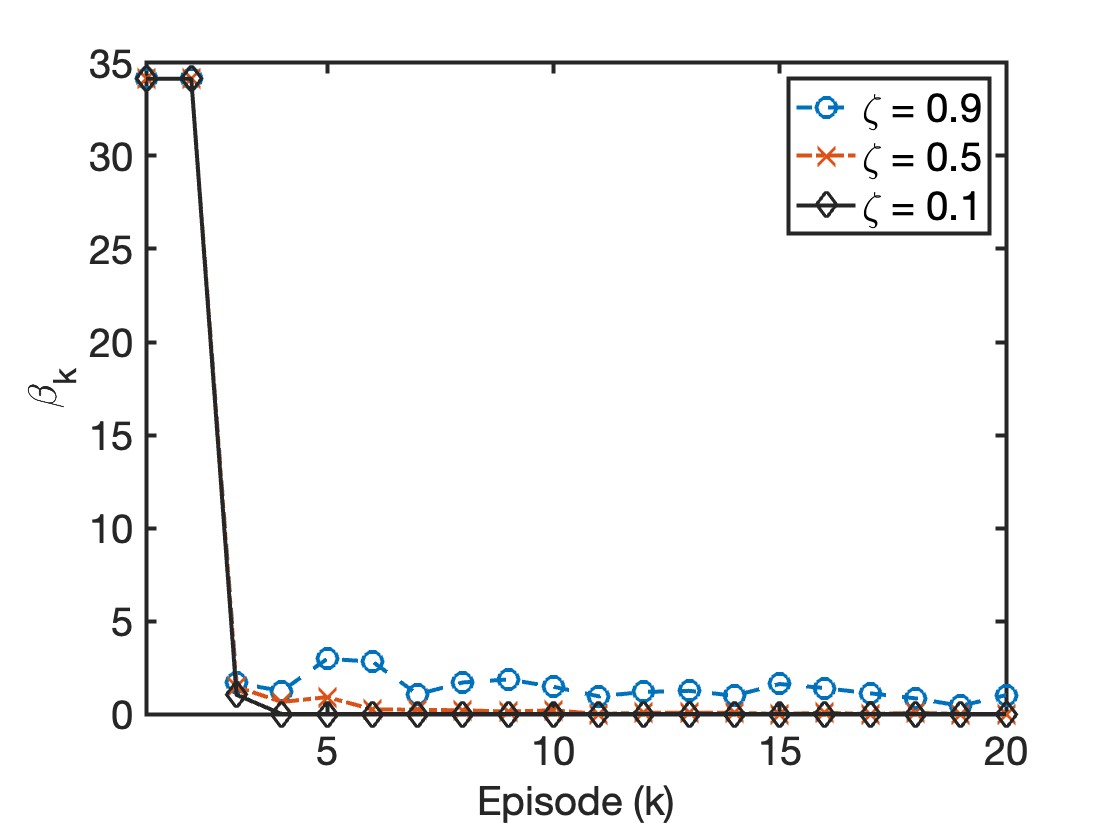}
\caption{{$\beta_k$ vs. Iteration $k$}}\label{fig:beta}
\end{figure}

\begin{figure}[t]
    \centering
\includegraphics[width=0.35\textwidth]{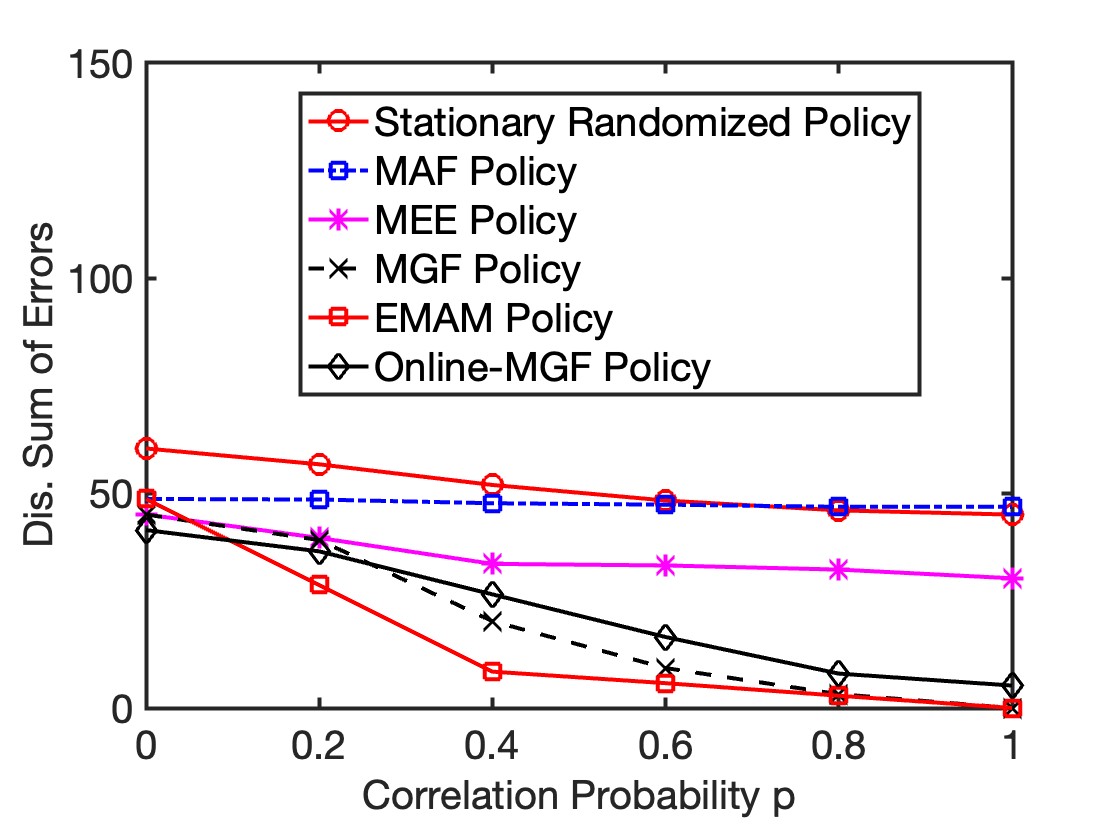}
\caption{Dis. Sum of Errors vs. $p$ with $M=10$ with $a^2_m=0.9$ in half of sources and other half with $a^2_m=0.7$. }\label{fig:prob1}
\end{figure}

\begin{figure}[t]
    \centering
\includegraphics[width=0.35\textwidth]{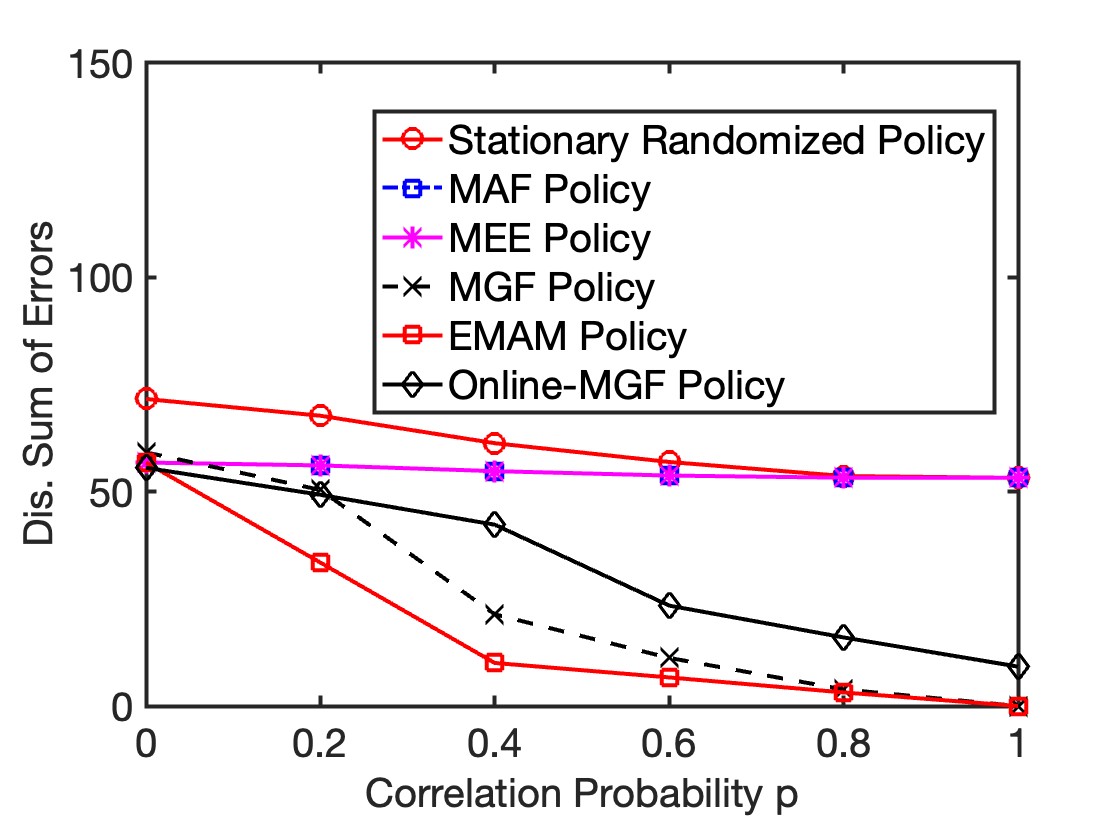}
\caption{Dis. Sum of Errors vs. $p$ with $M=10$ with $a^2_m=0.9$ in all sources.}\label{fig:prob2}
\end{figure}

{Now, we discuss the relationship of the correlation parameter $\epsilon_{m, n}$ and the correlation probability $p_{m, n}$. By definition \eqref{condition1}, we can write for the system discussed above as 
\begin{align}\label{epsilon}
    &\epsilon_{m, 1}^2=\nonumber\\
    &\max_{\delta_1, \delta_2, \ldots, \delta_M} I_L(Z_{m, t}; X_{1, t-1}|X_{1, t-\delta_1}, \ldots, X_{M, t-\delta_M}).
\end{align}
We can decompose the conditional mutual information as 
\begin{align}\label{CMIdecompose}
    &I_L(Z_{m, t}; X_{1, t-1}|X_{1, t-\delta_1}, \ldots, X_{M, t-\delta_M})\nonumber\\
    &=H_L(Z_{m, t}|X_{1, t-\delta_1}, \ldots, X_{M, t-\delta_M})\nonumber\\
    &~~~~~-p H_L(Z_{m, t}|Z_{m, t-1}, Z_{1, t-1}, X_{1, t-\delta_1},  \ldots, X_{M, t-\delta_M})\nonumber\\
    &~~~~~-(1-p)H_L(Z_{m, t}|Z_{1, t-1}, X_{1, t-\delta_1}, \ldots, X_{M, t-\delta_M})\nonumber\\
    &=p(H_L(Z_{m, t}|X_{1, t-\delta_1},\! \ldots, \!X_{M, t-\delta_M})\!\!-\!\!H_L(Z_{m, t}|Z_{m, t-1})),
\end{align}
where the last inequality holds because $Z_{m, t}$ and $Z_{1, t-1}$ are independent for all $m\neq 1$ and given $Z_{m, t-1}$, $Z_{m, t}$ is independent of $(X_{1, t-\delta_1}, \ldots, X_{M, t-\delta_M})$ for all $\delta_n\geq 1$. From \eqref{epsilon} and \eqref{CMIdecompose}, it implies that the correlation parameter $\epsilon_{m,1}$ increases as the correlation probability $p$ increases. When $p$ is zero, $\epsilon_{m, 1}$ is also zero. Also, $\epsilon_{m, 1}$ is maximized when $p$ is $1$.}

\subsection{Performance Evaluation}
In this section, we first evaluate the convergence of our ``Online Maximum Gain First" policy provided in Algorithm \ref{alg:gain}. Then, we compare the performance of following six policies. In our simulation, we consider $N=1$. 

\begin{itemize}[leftmargin=7mm]
\item {\textit{Exponential Moving Average Max-weight (EMAM) Policy:} The policy follows \cite[Algorithm 1]{tripathi2022optimizing}. The policy knows the probabilistic correlated structure within sources.} 
\item {\textit{Maximum Expected Error (MEE) Policy:} The policy is obtained from \cite{vishrantcorrelated}, where we use that fact that $W_{m, t}$ are independent over sources.}
\item \textit{Maximum Age First (MAF) Policy:} The policy schedules the source with highest AoI value. 
\item \textit{Stationary Randomized Policy:} 
Randomized policy selects client $m$ with probability $\beta_m/\sum_{m=1}^M \beta_m$, for every client $m$ with $\beta_m>0$. We consider $\beta_m=1$.
\item {\textit{Maximum Gain First (MGF) Policy:} The policy follows Algorithm \ref{alg:MGF}. The policy has access to the penalty function $f_m(\delta)$ defined in \eqref{option1}.} 
\item \textit{Online Maximum Gain First (Online-MGF) Policy:} The policy follows Algorithm \ref{alg:gain}. At episode $k=0$, Algorithm \ref{alg:gain} applies MAF policy. 
\end{itemize}

\begin{figure}[t]
    \centering
\includegraphics[width=0.35\textwidth]{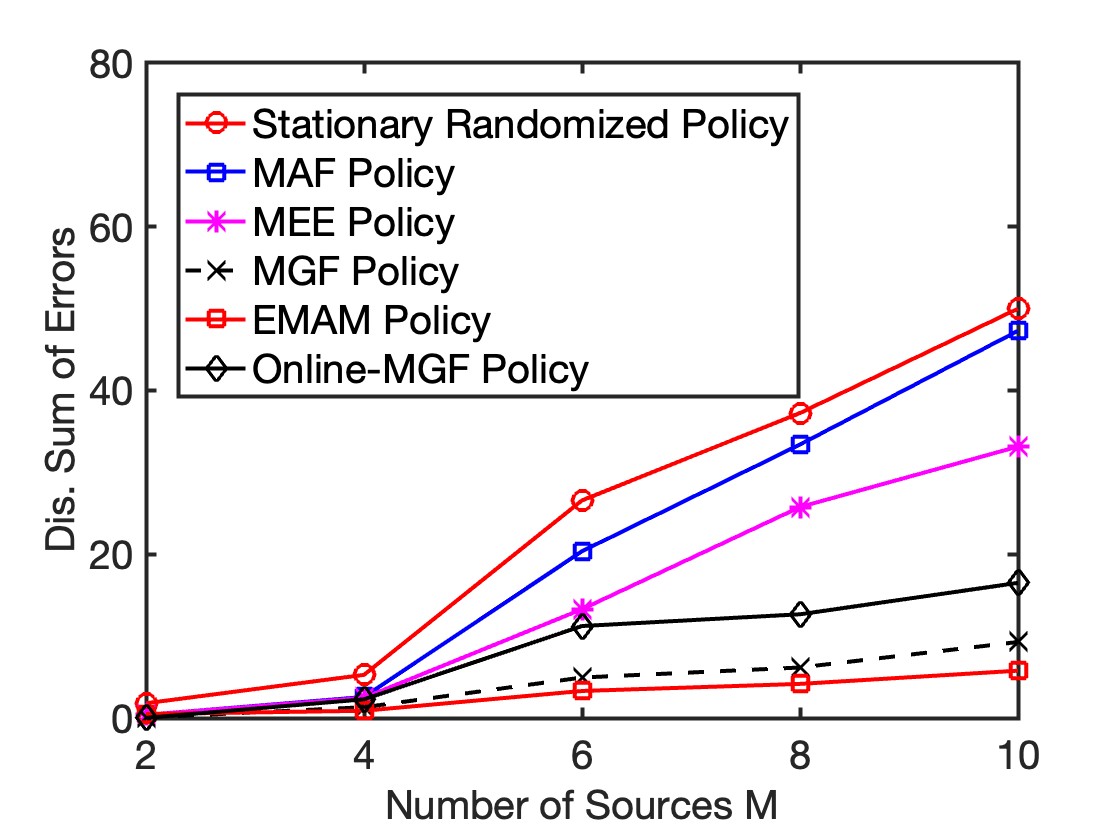}
\caption{Dis. Sum of Errors vs. $M$ with $p=0.6$}\label{fig:sources6}
\end{figure}

\begin{figure}[t]
    \centering
\includegraphics[width=0.35\textwidth]{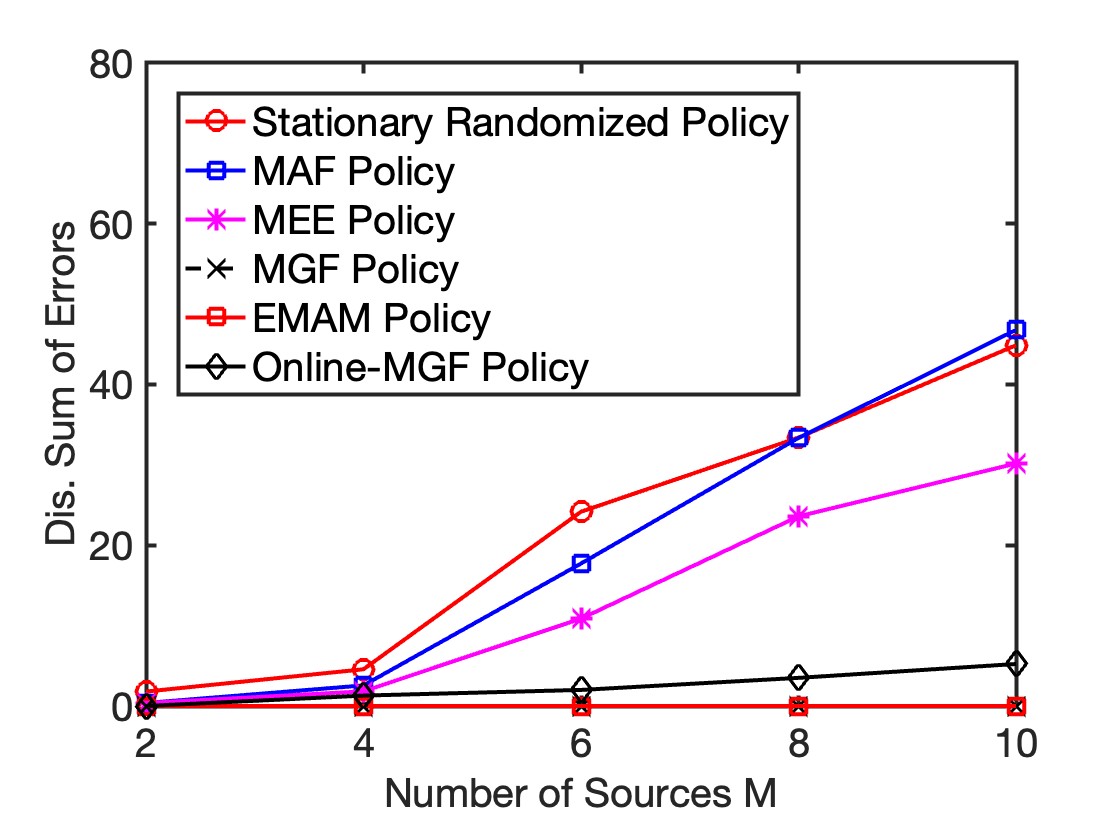}
\caption{Dis. Sum of Errors vs. $M$ with $p=1$}\label{fig:sources1}
\end{figure}

{Fig. \ref{fig:beta} illustrates the convergence behavior of our online-MGF policy. In Fig. \ref{fig:beta}, we plot $\beta_k$ defined in \eqref{difference} versus episode $k$, where $\beta_k$ is evaluated empirically $\beta_k=\max_{\delta} |\hat f_{m, {k-1}}(\delta)-\hat f_{m, {k}}(\delta)|$.
As $\beta_k$ converges to zero, policy $\pi_k$ converges. This is because the policy $\pi_k$ and $\pi_{k-1}$ are different only if $f_{m, k}$ and $f_{m, k-1}$ are different for a given Lagrange multiplier $\lambda$. From Fig. \ref{fig:beta}, we observe that smaller values of $\zeta$ ensures faster converges. 

Figs. \ref{fig:prob1}-\ref{fig:prob2} plot the discounted sum of inference errors versus the correlation probability $p$. The results show that the Online-MGF, MGF, and EMAM policies significantly outperform the MAF, Random, and MEE policies. EMAM becomes the best-performing policy, which is expected since it is specifically designed for this correlation model. A key advantage of the Online-MGF and MGF policies is their generality. Unlike the EMAM policy, which is tailored for a specific probabilistic correlation structure, our policies are derived without assuming any model and are thus applicable to any correlation type, but the policy EMAM can not be applied to other correlation structure. Despite this model-agnostic design, it is noteworthy that both MGF and Online-MGF perform close to EMAM. Furthermore, the Online-MGF policy performs comparably to the MGF policy without requiring knowledge of the penalty function. }

{Figs. \ref{fig:sources1}-\ref{fig:sources6} plot discounted sum of inference errors vs. number of sources $M$ with $p=0.6$ and $p=1$, respectively. Here, we set $a^2_m=0.9$ for half of the sources and $a^2_m=0.7$ for the other half. Similar to Figs. \ref{fig:prob1}-\ref{fig:prob2}, Figs. \ref{fig:sources1}-\ref{fig:sources6} also illustrate that the Online-MGF, MGF, and EMAM policies significantly outperform the MAF, Random, and MEE policies for any number of sources. As expected, EMAM achieves the best performance by leveraging its knowledge of the probabilistic correlation structure. When $p=1$, both the MGF and EMAM policies achieves the lowest inference error possible by scheduling the source with correlated updates.}

\section{Conclusion and Future Work}
{This paper investigated AoI-based signal-agnostic scheduling policies for correlated sources to minimize inference errors for multiple time-varying targets. We developed scheduling policies for two scenarios: (i) when the scheduler has full knowledge of the penalty functions and (ii) when the scheduler has no knowledge of these functions. In the first scenario, we proposed an MGF policy and we showed an upper bound of the optimality gap in an asymptotic region. In the second scenario, we first developed an online threshold-based scheduling policy for a relaxed version of the problem. We then leveraged the structure of this policy to design a novel online-MGF policy for the original problem. Simulation results demonstrate that our online policy effectively learns and exploits the correlation structure and system dynamics to achieve significant performance gains compared to Maximum Age First and random policies.}

Signal-aware scheduling problem is an important future direction, where the centralized scheduler has access to the signal values of the sources. While existing research \cite{orneeTON2021, SunTIT2020, OrneeMILCOM} has demonstrated that signal-aware scheduling can improve performance compared to AoI-based scheduling, it often comes at the cost of increased complexity. In contrast, AoI-based scheduling offers the advantages of low complexity and privacy preservation, which are important considerations in many applications. Moreover, distributed scheduling of correlated sources is an interesting future direction. 

\bibliographystyle{IEEEtran}
\bibliography{refshisher}

\begin{appendices}

\section{How to Compute Penalty Function $f_m(\delta)$}\label{approximatefunction}

It is possible to get closed form expressions of $f_m(\delta_m)$ for some well known random processes. An example is provided below:

\begin{example}
    Let observation of $m$-th source $X_{m, t}$ and the target $Y_{m,t}=X_{m,t}$ evolves as follows:
\begin{align}
    X_{m, t}&=a_m X_{m, t-1}+W_{m, t},
\end{align}
where the noise vector $\mathbf W_t=[W_{1,t}, W_{2,t}, \ldots, W_{M,t}]$ is an i.i.d. multi-variate normal random variable across time, i.e., $\mathbf W_t \sim \mathcal N(0, Q)$ and $Q$ is the covariance matrix. Let $q_{i,j}$ denote the $(i,j)$-th element of the covariance matrix $Q$. 
For this example, if $L(y, \hat y)=(y-\hat y)^2$ is a quadratic loss, then we have 
\begin{align}\label{derviationLTI}
\!\!  f_m(\delta)=
  \begin{cases}
       \bar q_{m,m} \delta, &\text{if}~a_m=1,\\
       \bar q_{m,m} \frac{a_m^{2\delta}-1}{a_m^2-1}, &\text{otherwise,}\\
    \end{cases}
    \end{align} 
    where $$\bar q_{m,m}=q_{m,m}-\mathbf q_m Q^{-1}_{-1} \mathbf q_m^T$$, $$\mathbf q_m=[q_{m,1}, \ldots, q_{m,m-1}, q_{m, m+1}, \ldots, q_{m, M}]$$ is the covariance of the $m$-th noise process with other process, and $Q^{-1}_{-1}$ is the noise covariance sub-matrix of the other processes excluding source $m$. 
\end{example}

{\bf Empirical Estimate of $f_m(\delta)$}: We can also estimate the penalty function $f_m(\delta)$ by using dataset collected in offline. 
\begin{example}
  Let $\{(Y_{m, t}, X_{m, t})_{m=1}^M\}_{t=1}^n$ be the dataset that contains $n$ samples of signal values from time $t=1$ to $t=n$. Then, the empirical estimate of $f_m(\delta)$ is given by
  \begin{align}
      f_m(\delta)=\frac{\sum_{t=\delta+1}^n L(Y_{m, t}, \phi_m(X_{m, t-\delta}, Z_{-m, t-1}))}{n-\delta}.
  \end{align}
\end{example}

\section{Proof of Lemma \ref{approx}}\label{papprox}

For the simplicity of presentation, we prove Lemma \ref{approx} for $M=3$ and $m=1$.

Part (a): we have 
\begin{align}\label{l3}
    g_1(\delta_1, \delta_2, \delta_3)&=H_L(Y_{1,t}|X_{1, t-\delta_1},X_{2, t-\delta_2}, X_{3, t-\delta_3})\nonumber\\
    &\geq H_L(Y_{1,t}|X_{1, t-\delta_1},X_{2, t-1}, X_{3, t-\delta_3})\nonumber\\
    &\geq H_L(Y_{1,t}|X_{1, t-\delta_1},X_{2, t-1}, X_{3, t-1})\nonumber\\
    &=f_m(\delta_1),
\end{align} 
where the inequalities hold due to \cite[Lemma 12]{dawid1998}. 

Part (b): We have 
\begin{align}\label{l3}
    g_1(\delta_1, \delta_2, \delta_3)&=H_L(Y_{1,t}|X_{1, t-\delta_1},X_{2, t-\delta_2}, X_{3, t-\delta_3})\nonumber\\
    &=H_L(Y_{1,t}|X_{1, t-\delta_1},X_{2, t-\delta_2}, X_{3, t-\delta_3})\nonumber\\
    &+H_L(Y_{1,t}|X_{1, t-\delta_1},X_{2, t-\delta_2}, X_{3, t-\delta_3}, X_{2, t-1})\nonumber\\
    &\quad-H_L(Y_{1,t}|X_{1, t-\delta_1},X_{2, t-\delta_2}, X_{3, t-\delta_3}, X_{2, t-1})\nonumber\\
    &\overset{(a)}{=}H_L(Y_{1,t}|X_{1, t-\delta_1},X_{2, t-1}, X_{3, t-\delta_3})\nonumber\\
    &+H_L(Y_{1,t}|X_{1, t-\delta_1},X_{2, t-\delta_2}, X_{3, t-\delta_3})\nonumber\\
    &-H_L(Y_{1,t}|X_{1, t-\delta_1},X_{2, t-\delta_2}, X_{3, t-\delta_3}, X_{2, t-1})\nonumber\\
    &\overset{(b)}{=}H_L(Y_{1,t}|X_{1, t-\delta_1},X_{2, t-1}, X_{3, t-\delta_3})\nonumber\\
    &+I_L(Y_{1, t}; X_{2, t-1}|X_{1, t-\delta_1}, X_{2, t-\delta_2}, X_{3, t-\delta_3})\nonumber\\
    &\overset{(c)}{=}H_L(Y_{1,t}|X_{1, t-\delta_1},X_{2, t-1}, X_{3, t-1})\nonumber\\
    &+I_L(Y_{1, t}; X_{3, t-1}|X_{1, t-\delta_1}, X_{2, t-1}, X_{3, t-\delta_3})\nonumber\\
    &\quad+I_L(Y_{1, t}; X_{2, t-1}|X_{1, t-\delta_1}, X_{2, t-\delta_2}, X_{3, t-\delta_3}),\nonumber\\
    &\overset{(d)}{\leq} f_m(\delta_1)+\epsilon_{1, 3}^2+\epsilon_{1,2}^2\leq f_m(\delta_1)+2\max_{n\neq 1}\epsilon_{1,n}^2\nonumber\\
    &\leq f_m(\delta_1)+O(\max_{n\neq 1}\epsilon_{1,n}^2),
\end{align} 
where (a) holds due to Markov chain $Y_{1, t} \leftrightarrow X_{2, t-1}\leftrightarrow X_{2, t-\delta_2}$ \cite[Lemma 12]{dawid1998}, (b) is obtained by definition \eqref{CMI}, (c) can be obtained by using similar steps of (a)\&(b), and (d) holds because of \eqref{condition1}. This completes the proof for $M=3$ and $m=1$. 

Following the same steps, we can prove it for any values of $M$ and $m$.

\section{Proof of Theorem \ref{approximationBound}}\label{papproximationBound}
We define 
    \begin{align}
        \Pi_c=\bigg\{\pi \in \Pi: \sum_{m=1}^M\pi_m(t)=N, \forall t=0, 1, \ldots \bigg\}.
    \end{align}
    Then, from \eqref{Multi-scheduling_problem1}-\eqref{Scheduling_constraint2} and \eqref{Multi-scheduling_problem2}-\eqref{Scheduling_constraint3}, we have
    {\blue \begin{align}\label{Theo1eq1}
        &\mathrm{V}_{opt}(T) \nonumber\\
        &\inf_{\pi \in \Pi_c}\mathbb{E}_{\pi} \left[\sum_{t=0}^{T-1}\gamma^t \sum_{m=1}^M g_m(\Delta_1(t), \Delta_2(t), \ldots, \Delta_M(t))\right]\nonumber\\
        &\overset{a}{=}\inf_{\pi \in \Pi_c}\mathbb{E}_{\pi} \left[\sum_{t=0}^{T-1}\gamma^t \sum_{m=1}^M f_m(\Delta_m(t))\right]\nonumber\\
        &+\sum_{t=0}^{T-1}\gamma^{t} \sum_{m=1}^M O\bigg(\max_{n=\neq m}\epsilon_{m,n}^2\bigg)\nonumber\\
        &=\mathrm{V}_{f, opt}(T)+\frac{1-\gamma^T}{1-\gamma}\sum_{m=1}^MO\bigg(\max_{n\neq m}\epsilon_{m, n}^2\bigg),
    \end{align}}
where (a) is due to \eqref{lowerbound}.

{\blue Moreover, we can write 
\begin{align}\label{Theo1eq2}
    &\mathrm{V}_{\pi_{f, opt}}(T)\nonumber\\
    &=\mathbb{E}_{\pi_{f, opt}} \left[\sum_{t=0}^{T-1}\gamma^t \sum_{m=1}^M g_m(\Delta_1(t), \Delta_2(t), \ldots, \Delta_M(t))\right]\nonumber\\
        &\overset{b}{=} \mathbb{E}_{\pi_{f, opt}} \left[\sum_{t=0}^{T-1}\gamma^t \sum_{m=1}^M f_m(\Delta_m(t))\right]\nonumber\\
        &+\sum_{t=0}^{T-1}\gamma^{t} \sum_{m=1}^M O\bigg(\max_{n=\neq m}\epsilon_{m,n}^2\bigg)\nonumber\\
        &=\mathrm{V}_{f, opt}(T)+\frac{1-\gamma^T}{1-\gamma}\sum_{m=1}^MO\bigg(\max_{n\neq m}\epsilon_{m, n}^2\bigg),
\end{align}
where (b) is due to \eqref{lowerbound}. 

By combining \eqref{Theo1eq1} and \eqref{Theo1eq2}, we obtain \eqref{theo1}.

\section{Proof of Theorem \ref{theorem2}}\label{ptheorem2}
By using triangle inequality, We can write
\begin{align}\label{3eq1}
    |\mathrm{V}_{\mathrm{gain}}(\boldsymbol{\lambda^*}, T)-\mathrm{V}_{opt}(T)|&\leq  |\mathrm{V}_{\mathrm{gain}}(\boldsymbol{\lambda^*}, T)-\mathrm{V}_{f,\mathrm{gain}}(\boldsymbol{\lambda^*}, T)|\nonumber\\
    &+|\mathrm{V}_{f,\mathrm{gain}}(\boldsymbol{\lambda^*}, T)-V_{f, opt}(T)|\nonumber\\
    &+|V_{f, opt}(T)-V_{opt}(T)|, 
\end{align}
where we define 
\begin{align}
  \mathrm{V}_{f,\mathrm{gain}}(\boldsymbol{\lambda^*}, T)=  \mathbb{E}_{\pi_{gain}} \left[\sum_{t=0}^{T-1}\gamma^t \sum_{m=1}^{M} f_m(\Delta_m(t))\right].
\end{align}
From \eqref{Theo1eq1}, we obtain
\begin{align}\label{3eq2}
    |V_{f, opt}(T)-V_{opt}(T)|=\frac{1-\gamma^T}{1-\gamma}\sum_{m=1}^{M} O\bigg(\max_{n\neq m}\epsilon_{m, n}^2\bigg)
\end{align}

Similar to \eqref{Theo1eq1}, we can also get 
\begin{align}\label{3eq3}
    &\mathrm{V}_{\mathrm{gain}}(\boldsymbol{\lambda^*, T})\nonumber\\
    &=\mathbb{E}_{\pi_{gain}} \left[\sum_{t=0}^{T-1}\gamma^t \sum_{m=1}^{M} g_m(\Delta_1(t), \Delta_2(t), \ldots, \Delta_M(t))\right]\nonumber\\
    &{=}\mathbb{E}_{\pi_{gain}} \left[\sum_{t=0}^{T-1}\gamma^t \sum_{m=1}^{M} f_m(\Delta_m(t))\right]\nonumber\\
    &+\sum_{t=0}^{T-1}\gamma^{t} \sum_{m=1}^{M} O\bigg(\max_{n=\neq m}\epsilon_{m,n}^2\bigg)\nonumber\\
    &=\mathrm{V}_{f,\mathrm{gain}}(\boldsymbol{\lambda^*, T})+\frac{1-\gamma^T}{1-\gamma}\sum_{m=1}^{M} O\bigg(\max_{n\neq m}\epsilon_{m, n}^2\bigg).
\end{align}

For an optimal tie-breaker function $\psi$ the gain index-based policy satisfies \cite[Proposition 5]{brown2020index}: 
\begin{align}\label{3eq4}
    \mathrm{V}_{f,\mathrm{gain}}(\boldsymbol{\lambda^*}, T)-V_{f, opt}(T)&\leq \mathrm{V}_{f,\mathrm{gain}}(\boldsymbol{\lambda^*}, T)-\mathrm{V}_{f, opt} (\boldsymbol{\lambda^*}, T)\nonumber\\
    &\leq \sum_{t=0}^{T-1}\beta_t\sqrt{N(1-r)},
\end{align}
where $\beta_t=\frac{\gamma^{t}}{2\gamma-1}((2\gamma)^{T-t}-1)$ as provided in \cite[equation (24)]{brown2020index}.
Next, we have \cite[equation (59)]{brown2020index} 
\begin{align}
    \sum_{t=0}^{T-1}\beta_t=\frac{1}{2\gamma-1}[2\gamma^{T}(2^{T}-1)-\frac{1-\gamma^{T}}{1-\gamma}].
\end{align}
For $\delta>0.5$, we have ${2\gamma-1}>0$. Then, we can drop the negative terms 
$$-\frac{2\gamma^{T}}{2\gamma-1}-\frac{1-\gamma^{T}}{1-\gamma}$$ and get
\begin{align}\label{3eq4}
    &\mathrm{V}_{f,\mathrm{gain}}(\boldsymbol{\lambda^*}, T)-V_{f, opt}(T)\nonumber\\
    &\leq O(\beta \sqrt{N(1-r)})= O(\beta \sqrt{Mr(1-r)}),
    \end{align}
where $\beta=\frac{2(2\gamma)^{T}}{2\gamma-1}$.}

By substituting \eqref{3eq2}-\eqref{3eq4} into \eqref{3eq1}, we obtain Theorem \ref{theorem2}.

\section{Proof of Theorem \ref{convergence}}\label{convergenceproof}


Part (a): From \cite{boyd2003subgradient}, we can say that our sub-gradient ascent method ensures the convergence of $\lambda_k$. Let $\lambda_k$ converge to $\lambda$. Next, as $k$ goes to $\infty$, $\zeta^k$ goes to $0$. Hence, from \eqref{estimatevalue} we can say that $$\lim_{k \to \infty} \tilde J_{m, \lambda_k, k}= \tilde J_{m, \lambda_{k-1}, k-1}.$$
This proves Theorem \ref{convergence}(a).

Part (b) To prove Theorem \ref{convergence}(b), the following Lemma is important.
   {\blue\begin{lemma}\label{prob}
  For the penalty function $f_{m, k}$, we have \begin{align}\mathrm{Pr}\bigg(f_{m, k}(\delta)\in B_{m,k}(\delta)\bigg)=1-\eta.
    \end{align}
\end{lemma}}
Lemma \ref{prob} can be directly proven by using the following Hoeffding's inequality.

\begin{lemma}[\textbf{Hoeffding's inequality}]
        Let $Z_1, Z_2, \ldots, Z_n$ are i.i.d. samples of a random variable $Z \in [0, 1].$ For any $\eta>0$, we must have 
        \begin{align}
            \mathrm{Pr}\bigg\{\bigg|\mathbb E[Z] -\frac{1}{n}\sum_{i=1}^n Z_i\bigg|\leq \sqrt{\frac{\mathrm{ln} (2/\bar \eta)}{2n}}\bigg\}\geq 1-\eta.
        \end{align}
    \end{lemma}

Now, using Lemma \ref{prob}, \eqref{estimateV}-\eqref{actualvalue}, we prove Theorem \ref{convergence}(b). By using triangle inequality, we can have 
\begin{align}\label{4eq2}
    &\bigg|J_{m,\lambda, {k}}(\delta)-J_{m,\lambda, {k-1}}(\delta)\bigg|\nonumber\\
    &\leq|J_{m,\lambda, {k}}(\delta)-J_{m, \mathrm{optimistic}, k}(\delta)|\nonumber\\
    &+ |J_{m, \mathrm{optimistic}, k}(\delta)-J_{m,\lambda, {k-1}}(\delta)|.
\end{align}
By using \cite[Theorem 12]{csaji2008value}, we can have with probability $1-\eta$: 
\begin{align}
    |\tilde f_{m, k-1}(\delta)-f_{m, k-1}(\delta)|\leq d_{m, k-1}(\delta).
\end{align}
This yields:

\begin{align}\label{4eq3}
    &|J_{m, \mathrm{optimistic}, k}(\delta)-J_{m,\lambda, {k-1}}(\delta)|\nonumber\\
    &\leq \frac{1}{1-\gamma}\max_{\delta'}|\tilde f_{m, k-1}(\delta')-f_{m, k-1}(\delta')|\nonumber\\
    &{\leq}\frac{1}{1-\gamma}\max_{\delta'}d_{m, k-1}(\delta')
\end{align}
and 
\begin{align}\label{4eq4}
    &|J_{m,\lambda, {k}}(\delta)-J_{m, \mathrm{optimistic}, k}(\delta)|\nonumber\\
    &\leq \frac{1}{1-\gamma}\max_{\delta'}|f_{m, k}(\delta')-\tilde f_{m, k-1}(\delta')|\nonumber\\
    &\leq \frac{1}{1-\gamma}\max_{\delta'}|f_{m, k}(\delta')-f_{m, k-1}(\delta')|\nonumber\\
    &+\frac{1}{1-\gamma}\max_{\delta'}d_{m, k-1}(\delta').
\end{align}
By substituting \eqref{4eq4}, \eqref{4eq3} into \eqref{4eq2}, we obtain Theorem \ref{convergence}(b).

{\blue \section{Minimization of Time-averaged Sum of Inference Errors}\label{avg}

Our goal is to find a policy $\pi \in \Pi$ that minimizes the time-averaged sum of inference errors: 
\begin{align}\label{Multi-scheduling_problemaverage}
&\mathrm{L_{opt}} =\nonumber\\
&\inf_{\pi \in \Pi}\limsup_{T\to\inf}\mathbb{E}_{\pi} \frac{1}{T}\left[\sum_{t=0}^{T-1} \sum_{m=1}^Mg_m(\Delta_1(t), \Delta_2(t), \ldots, \Delta_M(t))\right], \\\label{Scheduling_constraint1average}
&~\mathrm{s.t.} \sum_{m=1}^{M} \pi_{m}(t)=N, t=0, 1,\ldots, 
\end{align}
where $g_m(\Delta_1(t), \Delta_2(t), \ldots, \Delta_M(t))$ is the inference error for the $m$-th target at time $t$ and $\mathrm{L_{opt}}$ is the minimum average inference error.

For $M=2$ sources and $N=1$, we are able to find a low-complexity optimal policy to \eqref{Multi-scheduling_problemaverage}-\eqref{Scheduling_constraint1average}. According to Theorem \ref{theorem1}, the optimal policy is a stationary cyclic policy. 

\begin{definition}[\textbf{Stationary Cyclic Policy}]
A stationary cyclic policy is a stationary policy that cycles through a finite subset of points in the state space, repeating a fixed sequence of actions in a particular order.
\end{definition}


\begin{theorem}\label{theorem1}
For $M=2$ source case and $N=1$, there exists a stationary cyclic policy that is optimal for  \eqref{Multi-scheduling_problem1}-\eqref{Scheduling_constraint2}, where the policy consists of a period $\tau_1^*+\tau_2^*$ and in each period, source 1 is scheduled for $\tau_1^*$ consecutive time slots, immediately followed by source 2 being scheduled for $\tau_2^*$ consecutive time slots. The period $\tau_1^*$ and $\tau_2^*$ minimizes 
    \begin{align}\label{2opt}
    \mathrm{L_{opt}}=\min_{\substack{\tau_1=0,1,\ldots\\ \tau_2=0,1,\ldots}}\frac{1}{\tau_1+\tau_2}&\bigg(\sum_{k=0}^{\tau_1-1}\sum_{m=1}^2\big(g_m(1, 2+k)\big)\nonumber\\
    &+\sum_{j=0}^{\tau_2-1}\sum_{m=1}^2\big(g_m(2+j, 1)\big)\bigg),
\end{align}
where $\mathrm{L_{opt}}$ is the optimal objective value of \eqref{Multi-scheduling_problem1}-\eqref{Scheduling_constraint2}. 
\end{theorem}

\begin{proof}

Given $(\Delta_1(t), \Delta_2(t))=(\delta_1, \delta_2)$, we can show that there exists a stationary deterministic policy that satisfies the following Bellman optimality equation \cite{bertsekasdynamic1}:
\begin{align}\label{optimal2eqn1}
    h(\delta_1, \delta_2)&=\min_{(\pi_{1}(t), \pi_{2}(t)) \in \mathcal A'} g(\delta_1, \delta_2)- \mathrm{L_{opt}}\nonumber\\
    &+\pi_{1}(t) h(1, \delta_2+1)+\pi_{2}(t) h(\delta_1+1,1),
\end{align}
where $g(\delta_1, \delta_2)=\sum_{m=1}^2 g_m(\delta_1, \delta_2)$, $\mathcal A'=\{(0, 1), (1, 0)\}$, $h(\delta_1, \delta_2)$ is the relative value function for the state $(\delta_1, \delta_2)$, $\mathrm{L_{opt}}$ is the average inference error under an optimal policy, $(1, \delta_2+1)$ is the next state if source $1$ is scheduled, $(\delta_1+1, 1)$ is the next state if source $2$ is scheduled.

We can further express the Bellman equation \eqref{optimal2eqn1} for state $(\delta, 1)$ at time $t$ as follows: 
\begin{align}\label{optimal2eqn3}
   \!\!\! h(\delta, 1)\!\!=\!\!\min_{\tau_2\in \{0, 1, \ldots\}} \sum_{k=0}^{\tau_2}\bigg(g(\delta+k, 1)-\mathrm{L_{opt}}\bigg)+h(1,2),
\end{align}
where $\tau_2=0,1,\ldots$ is the time to keep scheduling source $2$ after time $t$. By solving \eqref{optimal2eqn3}, we get that the optimal $\tau_2(\delta)$ satisfies
\begin{align}\label{optimal2eqn4}
    \tau_2(\delta)=\inf \bigg\{\tau \in \mathbb Z^{+}: \gamma_1(\delta+\tau)\geq \mathrm{L_{opt}}\bigg\},
\end{align}
where $\gamma_1(\delta+\tau)$ is defined as 
\begin{align}\label{optimal2eqn5}
    \gamma_1(\delta)=\inf_{k=1,2,\ldots} \frac{1}{k}\sum_{j=0}^{k-1}g(\delta+1+j, 1).
\end{align}
Similarly, we can express the Bellman equation \eqref{optimal2eqn1} for state $(1, \delta)$ as follows:
\begin{align}\label{optimal2eqn6}
    \!\!\!\! h(1, \delta)\!\!=\min_{\tau_1\in \{0, 1, \ldots\}} \sum_{k=0}^{\tau_1}\bigg(g(1, \delta+k)-\mathrm{L_{opt}}\bigg)+h(2,1),
\end{align}
where $\tau_1=0,1,\ldots$ is the time to keep scheduling source $1$ after time $t$. By solving \eqref{optimal2eqn6}, we get that the optimal $\tau_1(\delta)$ satisfies
\begin{align}\label{optimal2eqn7}
    \tau_1(\delta)=\inf \bigg\{\tau \in \mathbb Z^{+}: \gamma_2(\delta+\tau)\geq \mathrm{L_{opt}}\bigg\},
\end{align}
where $\gamma_2(\delta+\tau)$ is defined as 
\begin{align}\label{optimal2eqn8}
    \gamma_2(\delta)=\inf_{k=1,2,\ldots} \frac{1}{k}\sum_{j=0}^{k-1}g(1, \delta+1+j).
\end{align}
Moreover, by using \eqref{optimal2eqn3} and \eqref{optimal2eqn6}, we get 

\begin{align}
    h(1,2)&= \sum_{k=0}^{\tau_1(2)}\bigg(g(1, 2+k)-\mathrm{L_{opt}}\bigg)+h(2,1)\nonumber\\
    &=\sum_{k=0}^{\tau_1(2)}\bigg(g(1, 2+k)-\mathrm{L_{opt}}\bigg)\nonumber\\
    &+\sum_{k=0}^{\tau_2(2)}\bigg(g(2+k, 1)-\mathrm{L_{opt}}\bigg)+h(1,2),
\end{align}
which yields
\begin{align}
    &\mathrm{L_{opt}}\nonumber\\
    &=\frac{\sum_{k=0}^{\tau_1(2)}\bigg(g(1, 2+k)\bigg)+\sum_{k=0}^{\tau_2(2)}\bigg(g(2+k, 1)\bigg)}{\tau_1(2)+\tau_2(2)}\nonumber\\
    &=\min_{\substack{\tau_1=0,1,\ldots\\ \tau_2=0,1,\ldots}}\frac{\sum_{k=0}^{\tau_1}\bigg(g(1, 2+k)\bigg)+\sum_{k=0}^{\tau_2}\bigg(g(2+k, 1)\bigg)}{\tau_1+\tau_2}.
\end{align}
Notice that the optimal objective $\mathrm{L_{opt}}$ is same as if source 1 is scheduled for $\tau^*_1$ consecutive time slots, followed by source 2 being scheduled for $\tau^*_2$ consecutive time slots.

\end{proof}

\begin{remark}
Theorem \ref{theorem1} presents a low-complexity optimal cyclic scheduling policy for two correlated sources. To our knowledge, this is the first such result. Prior work \cite{jhunjhunwala2018age} showed the existence of an optimal cyclic policy but required solving a computationally difficult minimum average cost cycle problem over a large graph. Our result provides a significantly more efficient solution for $M=2$ sources.
\end{remark}

For larger values of $M$, we can have a gain index-based policy similar to discounted version of the problem, but with average cost version. To get the gain index of the average cost version, we determine the value function as follows:

The optimal relative value function for each source $m$ is given by
\begin{align}\label{decomposedValue}
    &h_{m,\lambda}(\delta)=f_m(\delta)-\mathrm{f_{m, opt, \lambda}}+\min\bigg\{\lambda, h_{m,\lambda}(\delta+1)\bigg\},
\end{align}
where $\delta$ is the current AoI value and $h_{m,\boldsymbol{\lambda}}$ is the relative value function associated with the problem 
\begin{align}
&\mathrm{f_{m, opt, \lambda}} \!\!=\inf_{\pi_m \in \Pi_m}\! \limsup_{T\to \infty}\! \frac{1}{T}\mathbb{E}\! \left[\sum_{t=0}^{T-1} f_m(\Delta_m(t))+\lambda \pi_m(t)\right],  
\end{align}
with $h_{m, \lambda}(1)=0$, $f(\delta)$ is defined in \eqref{option1}. The relative value function can be obtained by using relative value iteration algorithm \cite{bertsekasdynamic1}. The average cost gain index is given by 
\begin{align}\label{gainindexavg}
\alpha_{m, \lambda}(\delta)=h_{m,\lambda}(\delta+1)-\lambda.
\end{align}

For the average cost problem, the Lagrange update rule becomes
\begin{align}
    \lambda_{k+1}=&\lambda_k+\frac{\theta}{k}\bigg(\frac{1}{T}\sum_{t=0}^{T-1}\sum_{m=1}^M \mathbf 1\bigg(\alpha_{m, \lambda_k}(\Delta_m(t))>0\bigg)-N\bigg).
\end{align}}

\end{appendices}


\end{document}